\documentclass{amsart}
\usepackage{amsmath}
\usepackage{amssymb}
\usepackage{paralist}
\usepackage[colorlinks=true]{hyperref}
\hypersetup{urlcolor=blue, citecolor=red}
\newtheorem{theorem}{Theorem}[section]
\newtheorem{proposition}[theorem]{Proposition}
\newtheorem{lemma}[theorem]{Lemma}

\theoremstyle{definition}
\newtheorem{remark}[theorem]{Remark}
\newtheorem{hypothesis}[theorem]{Hypothesis}

\numberwithin{equation}{section}
\numberwithin{theorem}{section}

\begin{document}
	
\title[A non-linear kinetic model of self-propelled particles]{A non-linear kinetic model of self-propelled particles with multiple equilibria}

\author[P.\ Butt\`a]{Paolo Butt\`a}
\address{Paolo Butt\`a\hfill\break \indent
	Dipartimento di Matematica, Sapienza Universit\`a di Roma 
	\hfill\break \indent
	P.le Aldo Moro 5, 00185 Roma, Italy}
\email{butta@mat.uniroma1.it}

\author[F. Flandoli]{Franco Flandoli}
\address{Franco Flandoli\hfill\break \indent 
	Scuola Normale Superiore of Pisa 
	\hfill\break \indent
	Piazza dei Cavalieri, 7	56126 Pisa, Italy}
\email{franco.flandoli@sns.it}

\author[M. Ottobre]{Michela Ottobre}
\address{Michela Ottobre\hfill\break \indent
	Heriot Watt University, Mathematics Department 
	\hfill\break \indent
	Edinburgh, EH14 4AS, UK}
\email{m.ottobre@hw.ac.uk}

\author[B. Zegarlinski]{Boguslaw Zegarlinski}
\address{Boguslaw Zegarlinski \hfill\break \indent
	Imperial College London 
	\hfill\break \indent
	South Kensington Campus, London SW7 2AZ}
\email{b.zegarlinski@imperial.ac.uk}

\begin{abstract}
	We introduce and analyse a continuum model for an interacting particle system of Vicsek type. The model is given by a non-linear kinetic partial differential equation (PDE) describing the time-evolution of the density $f_t$, in the single particle phase-space,  of a collection of interacting particles confined to move on the one-dimensional torus. The corresponding stochastic differential equation for the position and velocity of the particles is a conditional McKean-Vlasov type of evolution (conditional in the sense that the process depends on its own law through its own conditional expectation). In this paper, we study existence and uniqueness of the solution of the PDE in consideration. Challenges arise from the fact that the PDE is neither elliptic (the linear part is only {\em hypoelliptic}) nor in gradient form. Moreover, for some specific choices of the interaction function and for the simplified case in which the density profile does not depend on the spatial variable, we  show that the model exhibits multiple stationary states (corresponding to the particles forming a coordinated clockwise/anticlockwise rotational motion) and we study convergence to such states as well. Finally, we prove mean-field convergence of an appropriate $N$-particles system to the solution of our PDE:  more precisely, we show that the empirical measures of such a  particle system converge weakly, as $N \rightarrow \infty$, to the solution of the PDE. 
\end{abstract}



\maketitle
\thispagestyle{empty}

\section{Introduction}
\label{sec:1}

The emergence of coordinated movements and self-organization in the collective motion of large groups of individuals is ubiquitous in nature; e.g., in biological systems (flocking of birds, swarming of bacteria, etc.) or in human dynamics (movement of crowds, communications networks, etc.). Regardless of the great differences in the specific type of individuals, the collective motion of large groups of organisms exhibits similarities suggesting the existence of general rules that underlie collective dynamics. During the last two decades, this subject has drawn the attention of physicists and applied mathematicians, who developed several mathematical models, trying to extract the essential features of these phenomena, a first step towards the build of a full theory. 

In particular, in the pioneering work \cite{VCBCS}, the authors introduce the notion of self-propelled particles (SPP), which consist of locally interacting particles with an intrinsic driving force (internal to the particles, as it happens in real organisms). The model in \cite{VCBCS}, commonly known as the Vicsek model,  is given in two dimension, but similar SPP models exhibiting  self-organization were also introduced in one and three dimensions \cite{CBV,CV}. An explanation of these results by means of continuum theories have been proposed in \cite{CBV,TT1,TT2}, where hydrodynamic equations of motion for the density and velocity fields were used. Concerning the direct connection with real systems, we  recall, e.g., the interesting work \cite{Parisi}.

After then, many other models of self-propelled particles with pairwise interactions have been shown to exhibit coordinated motion, both  microscopic particle models \cite{PE,LRC,CDFSTB,CKF,CS,CS1,LLE} or macroscopic (fluid-like) models in the continuum \cite{PE,TB,TBL,EVL}. More recently, mesoscopic models based on kinetic equations for the distribution function in the one-particle phase space have been considered, see, e.g., \cite{HT,COP}, which represent a bridge between the particle- and the  hydrodynamic- description. Interesting developments are also contained in the works \cite{Ewelina, BCD,  Gianfelice, Mot}. The subject is impressively vast and a comprehensive literature is out of reach, so here and in the following we only mention the works that, to the best of our knowledge, are closest to the model treated in the present paper. For more precise literature review we refer the reader to the review papers \cite{Review1, Review2}. 
 
As already noticed, self-organization phenomena may occur also in one dimensional systems. In particular, in \cite{CBV} it is demonstrated that spontaneous symmetry breaking and ordering take place in one dimension as well. This is shown through the numerical analysis of a SPP model in one dimension, and heuristically justified by means of a continuum theory describing this kind of systems. It is worthwhile to quote also the interesting research \cite{locust}, where the behaviour predicted by this one dimensional SPP model is confirmed in experiments with marching locusts.

A mathematically rigorous explanation of these one dimensional phenomena as well as their continuum description is a natural query. In this paper, we introduce and study a kinetic model for a system of self-propelled particles moving in the one dimensional torus, which can be thought of as a mean-field version of the particle system introduced in \cite{CBV}. Related models have been considered in \cite{Rel1, Rel2, Rel3, Rel4}.

To write the kinetic evolution equation in consideration, some preliminary definitions and notation are needed. Letting $\mathbb{T}=\mathbb{R}/\mathbb{Z}$ be the one dimensional torus of length one, the phase space distribution function is denoted by $f_t = f_t(x,v)$, where $(x,v)\in \mathbb{T}\times \mathbb{R}$ and $t\in \mathbb{R}$ is the time. As customary in kinetic theory, the function $f_t(x,v)$ represents the density of particles which are at time $t$ in phase-space point $(x,v)$. A key ingredient in the definition of the dynamics is given by a function $G:\mathbb{R} \to \mathbb{R}$ such that 
\begin{equation}
\label{herdingG}
G(u) = - G(-u)\,, \qquad \left\{\begin{array}{ll} G(u) > u & \textrm{ if } 0<u<1\,, \\ G(u) < u & \textrm{ if } u>1\,. \end{array}\right. 
\end{equation}
This function $G$ incorporates both the propulsion and friction forces, whose effect is to set the average velocity to a prescribed value; further comments on the role of the function $G$ will be given after equation \eqref{SDE}.  The typical choice in the literature \cite{CBV,CV} is
\begin{equation}
\label{linG}
G(u) = \left\{\begin{array}{ll} (u+1)/2 & \textrm{ for } u>0\,, \\ (u-1)/2 & \textrm{ for } u<0\,, \end{array}\right. 
\end{equation}
or its generalization,
\[
G_\beta(u) = \frac{u+\beta\,\textrm{sign}(u)}{1+\beta}\,, \qquad \beta>0\,.
\]
The advantage of this choice is that $G_\beta'(u) = 1/(1+\beta)$ for any $u\ne 0$, but the above functions have a jump at $u=0$, rendering them less amenable to analysis.  Therefore, we will consider  smooth versions of the above functions which retain the property \eqref{herdingG}, see Hypothesis \ref{standingassumptions} and Hypothesis \ref{asslongtime} for precise assumptions; the regularization does not alter the asymptotic behavior of the resulting dynamics (however, see Remark \ref{statpoints} on the matter). We also introduce  a real-valued function $\varphi\in C^\infty(\mathbb{T})$ with the properties,
\begin{equation}
\label{normphi}
\varphi(x)\geq\epsilon>0\,, \qquad \varphi(x) = \varphi(-x)\,, \quad \int_{\mathbb{T}}\!\mathrm{d} x \, \varphi(x) = 1\,.
\end{equation}
Such a function will describe the interaction among particles (see few lines below). 
Finally,  for any phase space distribution function $f=f(y,w)$ and observable $F=F(y,w)$ we define,
\[
\langle F \rangle_{f,\varphi} (x):= \int_{\mathbb{T}}\!\mathrm{d} y \int_{\mathbb{R}}\!\mathrm{d} w\, f(y,w)\,\varphi(x-y)\, F(y,w)\,.
\]

With this notation in place, the evolution of  the system is governed by the following kinetic equation for the density $f_t$, 
\begin{equation}
\label{nl}
\partial_t f_t(x,v)=-  v\,\partial_x f_t(x,v) - \partial_v\big\{\big[G(M(t,x)) - v\big] f_t(x,v) \big\} + \sigma\,\partial_{vv}f_t(x,v)\,,
\end{equation}
where $\sigma>0$ is a given parameter and
\begin{equation}
\label{Mnl}
M(t,x) :=  \frac{\langle w \rangle_{f_t,\varphi}}{\langle 1 \rangle_{f_t,\varphi}} = \frac{\int_{\mathbb{T}}\!\mathrm{d} y \int_{\mathbb{R}}\!\mathrm{d} w\, f_t(y,w)\,\varphi(x-y)\, w}{\int_{\mathbb{T}}\!\mathrm{d} y \int_{\mathbb{R}}\!\mathrm{d} w\, f_t(y,w)\,\varphi(x-y)}\,.
\end{equation}
To stress the dependence of $M$ on $f$ we should write  $M_{f_t}(t,x)$ in place of $M(t,x)$. We refrain from doing so to avoid cumbersome notations.  
Since the total mass is conserved, we can assume that 
\begin{equation}
\label{normf}
\int_{\mathbb{T}}\!\mathrm{d} y \int_{\mathbb{R}}\!\mathrm{d} w\, f_t(y,w) = 1\,,
\end{equation}
provided this is true at time $t=0$. 

To better understand the meaning of the dynamics \eqref{nl}, let us  comment on the role of  all the addends in \eqref{nl}, starting from the linear part. If $G \equiv 0$ then the equation reduces to a linear Fokker-Planck equation for the hypoelliptic Ornstein-Uhlenbeck process, 
namely 
\begin{equation}\label{compnl}
\partial_t f_t(x,v)=-  v\,\partial_x f_t(x,v) + \partial_v(v f_t(x,v)) + \sigma\,\partial_{vv}f_t(x,v)\,.
\end{equation}
The above and related  equations are well studied in the literature (see the works \cite{Villani, Herau, Helferr, OttPavlPravda}); equation \eqref{compnl}  describes the evolution of the density of the law of a particle the dynamics of which is subject to  transport, damping and noise (the addends on the right-hand side, respectively). More explicitly, Eq.~\eqref{compnl} is the Fokker-Plank equation for the following Langevin-type dynamics
\begin{equation}\label{lang}
\begin{cases}
\mathrm{d} x_t = v_t \mathrm{d} t\,, \\ 
\mathrm{d} v_t = - v_t \mathrm{d} t +  \sqrt{2 \sigma} \mathrm{d} W_t\,,
\end{cases}
\end{equation}
with $W_t$ the standard one-dimensional Brownian motion. The measure $Z^{-1} e^{-v^2/2\sigma}$, $Z$ being the normalization constant, is the unique invariant measure for the SDE \eqref{lang}.  Let us now come to explain the nature of the nonlinearity. To this end, it may be useful to point out that  the solution $f_t(x,v)$  of \eqref{nl}  can be interpreted as the joint law of the  McKean-Vlasov process $(x_t, v_t) \in \mathbb{T} \times \mathbb{R}$ satisfying the nonlinear SDE,
\begin{equation}
\label{SDE}
\begin{cases}
\mathrm{d} x_t = v_t \mathrm{d} t\,, \\ \mathrm{d} v_t = \big[ G(M(t,x_t)) - v_t\big] \mathrm{d} t +  \sqrt{2 \sigma} \mathrm{d} W_t\,,
\end{cases}
\end{equation}
with
\begin{equation}\label{condmean}
M(t,x) = \frac{\mathbb{E} [v_t'  \, \varphi(x-x_t')] }{\mathbb{E} \,  \varphi(x-x_t')}\,,
\end{equation}
where $(x_t',v_t')\in \mathbb{T} \times \mathbb{R}$ denotes a random variable with law equal to the joint law of $(x_t,v_t)$, i.e., with probability density $f_t$ (again, a better notation for $M(t,x)$ should be $M_{f_t}(t,x)$). We emphasize that, in view of \eqref{condmean}, the evolution \eqref{SDE} is non-linear in the sense of McKean, as the process depends on its own law through the expectation in \eqref{condmean}. Furthermore, from the expression \eqref{condmean} one can see that $M(t,x)$ is a sort of ``local average velocity'' around $x$, weighed with the auxiliary function $\varphi$. For example, take $\varphi$ to be the characteristic function of the interval $[-b,b]$ (this choice does not satisfy the assumptions \eqref{normphi}, but it is meaningful for expository purposes): in this case, $M(t,x)$ reduces to the  conditional expectation $\mathbb{E} (v_t' \vert x_t' \in [x-b, x+b])$ with $(x_t',v_t')$ distribuited as $(x_t,v_t)$. The role of the nonlinear force $G(M(t,x_t)) - v_t$  is then quite clear: it tends to increase (decrease, respectively) the velocity of $x_t$ if the latter is larger (smaller, respectively) of the local average velocity $M(t,x_t)$ around $x_t$. In other words, while the function $\varphi$ describes the spatial interaction, the function $G$ has the role of ``herding'' the particles towards a certain average velocity. This will be more clear after the results of Section \ref{sec:2},  see in particular equation \eqref{ODEm1}. 

In this paper we study various aspects of the evolution \eqref{nl}.  Section \ref{sec:3} is devoted to the study of the well-posedness of the non-trivial Cauchy problem associated with \eqref{nl}  and in  Section \ref{sec:4}  we prove that an appropriate $N$-particle system converges, in the continuum limit $N \rightarrow\infty$, to the solution of the PDE \eqref{nl} .  Section  \ref{sec:2} and Section \ref{sec:5} contain results on invariant measures and long time behaviour, for some particular choices of the interaction function $\varphi$. 

Let us come to describe the main results of this paper in more detail. In Section \ref{sec:3} we show global (in time) existence and uniqueness for the Cauchy problem associated with \eqref{nl}. We adopt here, with non-trivial modifications,  the approach taken in \cite{Degond} to study well-posedness of the Vlasov-Fokker-Planck equation. While the methods we use in this section are purely analytical, alternative techniques, probabilistic in flavour, have been used to study similar problems, see \cite{Bossy}. 

As already observed, Eq.~\eqref{nl} has been introduced as a mean-field counterpart of the discrete model in \cite{CBV}. To make more precise the connection with particle systems, in Section \ref{sec:4} we rigorously show how the  kinetic evolution considered here can be derived as the mean-field limit of a suitable model of self propelling-particles with additive Brownian noises. In this framework, we recall that the bulk behavior of interacting diffusions and Ornstein-Uhlenbeck processes has been already considered in the case of forces depending solely on particle position, see, e.g., \cite{BL,OV,T,V, {Gomes}}. Here, we consider the following  system of $N$ interacting particles: for every $i\in \{1, \dots, N\}$, the position and velocity of the $i$-th particle, denoted by $x_t^{i,N}$ and $v_t^{i,N}$, respectively, evolve according to the SDE
\begin{align}
\mathrm{d} {x}_t^{i,N} & = v_t^{i,N} \, \mathrm{d} t\,,  \label{parsys1}\\\
\mathrm{d} {v}_t^{i,N} & = - v_t^{i,N} \mathrm{d} t + 
G\left(\frac{\frac 1N \sum_{j=1}^N \varphi (x_t^{i,N} - x_t^{j,N}) v_t^{j,N}}{\frac 1N \sum_{j=1}^N \varphi (x_t^{i,N} - x_t^{j,N})}\right) \mathrm{d} t + \sqrt{2\sigma} \mathrm{d} W_t^{i}\,,\label{parsys2}
\end{align}
where the $W^i$'s are one-dimensional independent standard Brownian motions. (Because $G$ is assumed to be bounded and Lipschitz continuous, existence and uniqueness (for every $N$ fixed) of the strong solution of the above system follows from standard SDE theory.)  The main result of Section \ref{sec:4} can then be expressed as follows (for a precise statement see Theorem \ref{teo:convpartsys}): as $N\rightarrow\infty$,  the motion of each particle converges weakly to the solution of the non-linear diffusion \eqref{SDE}.  To put it differently, the empirical measure of the system, i.e., the measure $S_t^N(\mathrm{d} x, \mathrm{d} v) := \frac 1N \sum_{i=1}^N \delta_{(x_t^{i,N}, v_t^{i,N})}(\mathrm{d} x, \mathrm{d} v)$, converges to a measure with a density with respect to the Lebesgue measure; such a density  is precisely the solution of the PDE \eqref{nl}. We notice in passing that the method of proof adopted here produces an existence and uniqueness result in measure space for the evolution \eqref{nl}, see Proposition \ref{uniquenessmeassol} for precise statement.

Regarding the long time behaviour, in Section \ref{sec:2} we start by considering the simplified setting in which the solution does not depend on the spatial-variable. That is, we only consider space-homogeneous evolutions $f=f_t(v)$. In this case we give a complete description of the asymptotic behaviour: we show that  the system admits three space-homogeneous stationary densities, which are Gaussian distributions with same variance $\sigma$ and average velocity $0$ and $\pm 1$, respectively. To understand why this is the case, one should observe that the points $0$ and $\pm 1$ are precisely the only points such that $G(x)=x$ and reconcile such an observation with the above  discussion regarding the nature of the non-linear term of the PDE \eqref{nl}.    In this framework, the Gaussians $\mathcal{N}(\pm 1, \sigma)$ correspond to clockwise/anticlockwise  coordinated motion;  one could argue that such measures  may be  the only  physically relevant equilibria -- in the sense that,  had we considered  $G$ as in \eqref{linG}, the mean zero Gaussian would not be an equilibrium, as the function $\eqref{linG}$ does not intersect the diagonal at  the origin,   see Remark \ref{statpoints} and Remark {\ref{rem:critpointsS}}. Still under the simplifying assumption $f_t=f_t(v)$, we prove that if the initial average velocity is positive, i.e., $\int_{\mathbb{R}}\mathrm{d} v\,  v\, f_0(v)>0$, then  the solution of the PDE converges to the Gaussian $\mathcal{N}(1, \sigma)$; if the initial profile $f_0$ has negative average velocity, then the solution converges to the Gaussian with negative mean, $\mathcal{N}(-1, \sigma)$.  This implies that, regardless of the initial conditions, the system is driven towards a coherent motion (in the context of flocking models this phenomenon is sometimes called {\em unconditional flocking}).  Moreover we prove that in either cases decay to equilibrium is exponential. We prove this fact by using two different approaches: in Section \ref{sec:2.1} we close equations on the cumulants of the distribution (closing equations directly on the moments did not seem possible); in Section \ref{sec:2.2} we identify a suitable entropy functional for the system and employ an appropriate modification of the Bakry-Carrillo-MacCann-Toscani-Villani strategy to show exponentially fast  decay of such a functional. To summarize, in this simpler case we can determine all the equilibria of the system as well as their basin of attraction and the rate of decay. This is a rather ``lucky case". Indeed, for dynamics with multiple equilibria, there is at present no general theory to address the problem of convergence (the primary problem here being the description of the basin of attraction of each invariant measure) - as opposed to the ergodic framework, where several approaches have been developed \cite{Oneeq1, Villani, CMV, Toscani, Oneeq2}. However, we flag up the work \cite{Cass}, where the authors characterize a class of SDEs with multiple equilibria for which one can give a complete description of the basin of attraction of each stationary state. We also mention the case of a particular nonlinear Fokker-Planck equation, introduced in \cite{BCP} as a kinetic model for a one-dimensional granular medium, which is characterized by a unique, but non-Maxwellian, steady distribution, see \cite{BCCP}. 

The analysis of the long-time behavior  for the full equation \eqref{nl} appears indeed much challenging (not last, because the the dynamics \eqref{nl} is not in gradient form) and it will be the content of future work.  In particular, in this fully non-linear and non-space-homogeneous case \eqref{nl}, one can observe by direct calculation that the three space-homogeneous Gaussians $\mathcal{N}(0,\sigma)$ and $\mathcal{N}(\pm1, \sigma)$ are still equilibria for the system (one can also determine an {\em evolution system of measures} for such a dynamics, see Remark \ref{rem:2.2}); but we have been unable to prove, in such generality, that these are indeed the only invariant measures for the system.  However, in Section \ref{sec:5} we show that if the interaction function is a ``perturbation of the constant function", i.e., if $\varphi(x)=1+\lambda \psi(x)$, where $0<\lambda \ll 1$ and $\psi$ is a mean zero, bounded function (defined on the one-dimensional torus), then the above three Gaussians are still the only invariant measures of the dynamics. 

\subsection{Notation and standing assumptions}
\label{sec:1.1}

Throughout the paper, we make the following assumptions regarding the interaction function $\varphi$ and the ``herding function" $G$. 
\begin{hypothesis}
\label{standingassumptions} 
The function $\varphi: \mathbb{T} \rightarrow \mathbb{R}$ and $G: \mathbb{R} \rightarrow \mathbb{R}$ appearing in \eqref{nl}-\eqref{Mnl} satisfy the following
\begin{itemize}
\item[(1)] The function $\varphi$ is smooth and satisfies \eqref{normphi}, for some $\epsilon>0$. 
\item[(2)] The function $G$ is smooth and bounded, with bounded derivatives of any order (hence globally Lipschitz). 
\end{itemize}
\end{hypothesis}
\begin{hypothesis}
\label{asslongtime}
$G(x)$ satisfies \eqref{herdingG} and   $G(x)=x$ only at the points $x=-1, 0, 1$. 
\end{hypothesis}
\begin{remark}
\label{statpoints}
Some comments on the above assumptions. 
\begin{itemize}
\item[(i)] Hypothesis \ref{standingassumptions} is assumed to hold throughout and it is used to establish matters of well-posedness. Hypothesis \ref{asslongtime} is only needed when we want to make statements about the long-time behavior of the dynamics. For this reason, in Section \ref{sec:2} and Section \ref{sec:5} we work under both Hypothesis \ref{standingassumptions} and Hypothesis \ref{asslongtime}. The rest of the paper hinges solely on Hypothesis \ref{standingassumptions}. 
\item[(ii)] For physical purposes, it would be interesting to consider interaction functions $\varphi$ with compact support. We don't do this here to avoid further technicalities in the proofs but we reserve to drop this assumption in future work.
\item[(iii)] Clearly, one could pick the function $G$ to intersect the diagonal in an arbitrary number of points. The choice $x=0,1,-1$ is purely arbitrary. If $G$ intersects the diagonal elsewhere, say at $a>0$  (so that by antisymmetry one also has  $G(-a)=-a$), then also the Gaussian measures $\mathcal{N}(a, \sigma)$  and  $\mathcal{N}(-a, \sigma)$ will be invariant for  the  reduced space-homogeneous equation considered in Section \ref{sec:2}  as well as for the full dynamics    \eqref{nl}. This can be seen with calculations completely analogous to those presented at the beginning of Section \ref{sec:2}, see also comments after \eqref{ODEm1}. 
\end{itemize} 
\end{remark}
Throughout, we will interchangeably use the notation $f_t$ and $f(t)$ for time-dependent functions. Further notation will be somewhat local to individual sections and we introduce it when needed. 

\section{The space-homogeneous case}
\label{sec:2}

Here, we study the space-homogeneous solutions of Eq.~\eqref{nl}. We work under Hypothesis \ref{standingassumptions} and Hypothesis \ref{asslongtime}.  In view of \eqref{normphi}, \eqref{Mnl},  and \eqref{normf}, in the space-homogeneous case $f=f_t(v)$ the kinetic equation reads,
\begin{equation}
\label{p1}
\partial_t f_t = - G( \langle w \rangle_{f_t}) \partial_v f_t +  \partial_v (vf_t) +  \sigma\,\partial_{vv}f_t(v)\,,
\end{equation}
where $\langle w\rangle_f $ denotes the average of the distribution $f$, that is
\begin{equation}
\label{mf}
\langle w\rangle_f := \int_{\mathbb{R}}\!\mathrm{d} w\, f(w)\, w\,.
\end{equation}
Later on, see Eq.~\eqref{ODEm1}, we will observe that the mean $M_1(t):=\langle w \rangle_{f_t}$ solves a simple ODE. 
Existence and uniqueness of classical solutions of the Cauchy problem associated with \eqref{p1} with initial datum $f_0(x)$ such that $\left\vert f_0(x)\right\vert\le c \exp(a v^2)$, $ a,c>0$,  is shown in \cite{ArBes}. Existence and uniqueness for \eqref{p1} in Sobolev spaces follows after a change of variable: if $f_t(v)$ is the solution to \eqref{p1} then the function
\begin{equation}
\label{chofvar}
h_t(v) := f_t\big(v- \int_0^t\!\mathrm{d} s\, G( M_1(s))\big)
\end{equation}
solves
\[
\partial_t h_t(v) =  \partial_v(vh_t) + \sigma \partial_{vv}h_t\,,
\quad h_0(v)=f_0(v)\,,
\]
which is the Fokker-Planck equation for the Ornstein-Uhlenbeck process, whose well posedness is well established (even for measure-valued solutions, see, e.g., \cite[Chap.1]{Stroock}).\footnote{Furthermore, this equation can also be seen as a special case of the more general setting considered in Proposition \ref{propwellposlin}.} For such a process it is well known that if the following assumption holds,

\smallskip
{\bf [H]} {\em The initial datum $f_0$ is such that $f_0$, $\left\vert f_0\right\vert^2$ and $\left\vert \partial_vf_0 
\right\vert^2$ have finite moments of any order}
\smallskip

\noindent
 then the solution at time $t>0$ enjoys the same property and in the reminder of this section we shall work under such an assumption on the initial profile (see \cite[Theorem 2.29 and other results within that section]{Nee} or \cite{Bakry}).\footnote{This assumption is certainly not sharp but sufficient to our purposes. Here we are not interested in minimal assumptions but in a description of the dynamics. For sharper conditions see the cited references. } We  do not linger on matters of well posedness and we move forward to  characterize the stationary distributions (invariant measures) of \eqref{p1}.    

After observing that $G( \langle w \rangle_{f})$ does not depend on the velocity variable, the invariant measures are the normalised solutions to the equation 
\[
\partial_v\left[  (-G( \langle w \rangle_{f}) +v) f +  \sigma\,\partial_{v}f(v\right]) = 0\,.
\]
Clearly, solving the above amounts to  solving  the first order ODE,
\[
\big[G(\langle w\rangle_f) - v\big] f(v) = \sigma\,f'(v) + C\,,
\]
where $C$ is a generic constant. 
Integrable positive solutions to the above ODE exist only for $C=0$ and have the form,
\[
f(v) = \bar{C} \, \exp\left\{\frac{2G(\langle w\rangle_f)v - v^2}{2\sigma}\right\}\,.
\]
The normalization condition $\int_{\mathbb{R}}\!\mathrm{d} w\, f(w) =1$ gives $\bar{C} = \exp (G^2(\langle w\rangle_f))/(\sqrt{2\pi\sigma})$. From  \eqref{mf}, one  also obtains the condition $G(\langle w\rangle_f) = \langle w\rangle_f$ which implies, by our choice of $G$, $\langle w\rangle_f = 0,\pm 1$. In conclusion, Eq.~\eqref{p1} admits three stationary solutions, namely the following three Gaussian densities,
\begin{equation}
\label{sm}
\mu_0(v) = \frac{1}{\sqrt{2\pi\sigma}}\,\exp\left[-\frac{v^2}{2\sigma}\right]\,, \qquad \mu_\pm(v) = \frac{1}{\sqrt{2\pi\sigma}}\,\exp\left[-\frac{(v\mp 1)^2}{2\sigma}\right]\,.
\end{equation}

\begin{remark} 
\label{rem:2.1}
In particular, the Gaussian densities $\mu_0$ and $\mu_\pm$ are the only space-homogeneous stationary densities of Eq.~\eqref{nl}. 
\end{remark}

The asymptotic behaviour of the time-dependent distributions can be fully characterized, we show below two different approaches.

\subsection{Cumulants}
\label{sec:2.1}

Working under the assumption {\bf [H]}, we set, for any integer $n \ge 0$, 
\[
M_n(t) := \langle w^n \rangle_{f_t} := \int_{\mathbb{R}}\!\mathrm{d} w\, w^n f_t(w)\,.
\]
(We will favour the use of the notation $M_n(t)$ when we wish to stress the time-dependence and the use of the notation $\langle w^n \rangle_{f_t}$ when we wish to stress the dependence on the function $f_t$.)
By an explicit computation we have
\begin{align}
\label{ODEm0} & \dot M_0 = 0\,, \\ 
\label {ODEm1} & \dot  M_1  = G(M_1)-M_1 \,, \\ \label{ODEm2} & \dot  M_2 = 2 [G(M_1)M_1 -M_2 + \sigma ]\,. 
\end{align}
The first equation is the conservation of total mass (so that $M_0(t) = 1$ for all $t\geq 0$) and by the last equation the variance converges to $\sigma$.  Eq.~\eqref{ODEm1} is instead a statement about the average velocity:   if the average velocity is positive (negative, respectively) at time $t=0$, then it converges to $+1$ ($-1$, respectively). Notice that in order for this fact to be true one just needs $G(M_1)-M_1$ to be positive in $(0,1)$ and $(-\infty, -1)$, negative in $(-1, 0)$ and $(1, \infty)$ (the antisymmetry of $G$ or other detailed properties of such a function don't really matter to this end). If the  function $G$ had  more than three fixed points, then, for every   $a\in \mathbb{R}$ such that $G(a)=a$, one would have an additional invariant measure, namely the Gaussian $\mathcal{N}( a, \sigma)$ (and, by antisymmetry of $G$, the measure $\mathcal{N}(-a, \sigma)$ as well). Whether the mean velocity ever converges to $a$ (or $-a$) would depend on the sign of $G(x)-x$ for $x>a$ and for $x>a$. 

For moments of any order, after multiplying \eqref{p1} by $v^n$ and integrating by parts,\footnote{In view of assumption {\bf [H]} all the needed integrations by parts are actually justified and give zero boundary terms.} one finds the recursive relation,
\begin{equation}
\label{eqn:recursionmoments}
\begin{split}
& \dot M_n = nG(M_1) M_{n-1}-n M_n + \sigma n (n-1) M_{n-2}\,, \quad  n \ge 1\,, \\ &  \text{(where $M_{-1}(t):=M_0(t)=1$)\,.}
\end{split}
\end{equation}
Let $C_n(t)$ denote the  $n$-th order cumulant of $f_t(v)$, namely
\[
C_n(t) = \frac{\mathrm{d}^n}{\mathrm{d} \lambda^n} \log \langle e^{\lambda w}\rangle_{f_t} \bigg|_{\lambda=0}\,.
\]
Then $C_0(t) = 0$, $C_1(t)=M_1(t)$, and $C_2(t)=M_2(t)-M_1^2(t)$, so that
\begin{align}
\label{ODEc1}
& \dot C_1 = G(C_1) - C_1\,, \\ \label{ODEc2} & \dot C_2  = -2 C_2+2 \sigma\,.
\end{align}
We also recall the following relation between moments and cumulants \cite{Smith},
\begin{equation}
\label{eqn:commomen}
C_{n}=M_{n}-\sum_{j=1}^{n-1} \binom{n-1}{j-1}M_j C_{n-j}\,.
\end{equation}

Convergence to the stationary solutions for the dynamics \eqref{p1} is now a consequence of the following lemma, whose proof is postponed to the end of this subsection.  

\begin{lemma}
\label{lemma:ODEforcn}
Assume {\bf [H]} holds. With the notation introduced so far,  we then have
\begin{equation}
\label{eqn:dtCn=-Cn}
\dot C_n(t)= -n C_n(t) \quad \forall\, n \ge   3\,. 
\end{equation}
\end{lemma}

\begin{proposition}
\label{cor:2.1}
Let Hypothesis \ref{standingassumptions} and Hypothesis \ref{asslongtime} hold. Let $f_t$ be the solution of \eqref{p1}. If the initial datum $f_0$ has positive mean, i.e., $M_1(0)>0$, and satisfies {\bf [H]}, then $f_t$ converges to the stationary state $\mu_+$ as $t \to \infty$. Analogously, if the initial datum has negative (zero, respectively) mean, then the solution of \eqref{p1} converges  to $\mu_-$ ($\mu_0$, respectively).
\end{proposition}

\begin{proof}[Proof of Proposition \ref{cor:2.1}]
By lemma \ref{lemma:ODEforcn} the cumulants of order $n \ge 3$ of the density $f_t$  tend to zero as $t \to \infty$. Therefore, the solution to the space-homogeneous Eq.~\eqref{p1} converges to a Gaussian density (it is a standard fact that the only distributions with vanishing cumulants of order $n\geq 3$ is the Gaussian distribution). Equations \eqref{ODEm1} and \eqref{ODEm2} describe the evolution of the mean and the variance of the solution of \eqref{p1}. As already observed, by \eqref{ODEm2} the limiting variance is $\sigma$. Because of our assumptions on $G$, by \eqref{ODEm1} the limiting mean is $\pm 1$ or zero if $M_1(0)$ is positive, negative or zero, respectively.  
\end{proof}

\begin{remark}[On the non space homogeneous dynamics \eqref{nl}]
\label{rem:2.2}
Another elementary consequence of Lem\-ma \ref{lemma:ODEforcn} together with \eqref{ODEc1} and \eqref{ODEc2} is the following. Suppose the initial datum for the dynamics \eqref{p1} is a Gaussian measure with given mean $M_0$ and variance $\sigma$. Then the law of the process at time $t$ is still Gaussian, with variance $\sigma$ and mean $M_1(s)$, where $M_1(s)$ is the solution to \eqref{ODEc1} with initial datum $M_0$. Otherwise stated, for each $M_0 \in \mathbb{R}$, the family of measures $\{\nu_t \sim \mathcal{N}(M_1(t), \sigma)\}_{t\geq 0}$ is an \textit{evolution system} of measures, see \cite{DPR, Lunardi}. 

Clearly, these families of measures constitute an evolution system even for the non-homogeneous dynamics \eqref{nl}. More generally, we can consider densities  $f_t(x,v)$ of the form, 
\begin{equation}
\label{ftpluginnl}
\begin{split}
& f_t(x,v) = \mathcal{Z}_t^{-1} \exp\bigg(-\frac{(v-A(t,x))^2}{2 B(t,x)} \bigg)\,, \quad 
\mathcal{Z}_t := \int_{\mathbb{T}}\! \mathrm{d} x\, \sqrt{2 \pi B(t,x)} \,, \\ & \mbox{ and } \,\, B(t,x)>0\,,
\end{split}
\end{equation}
and look for the class of (regular enough) functions $A=A(t,x)$ and $B=B(t,x)$ such that the density $f_t$ is a solution to \eqref{nl}. By substituting in \eqref{nl} one obtains a long expression, which can be rearranged to be a polynomial in $(v-A)$. By comparing the coefficients with equal power, one obtains the following system of constraints,
\begin{align}
&\partial_x B =0\,, \label{maybe1}\\
&\partial_t B + 2B \partial_x A + A \partial_xB-2\sigma+ 2B =0\,, \label{maybe2} \\
& \partial_tA + AB \partial_xA-G(M)+A=0\,, \label{maybe3}\\
&-\frac{\int_{\mathbb{T}} \mathrm{d} x\, B^{-1/2} \partial_t B}{2 \int_{\mathbb{T}} \mathrm{d} x\, \sqrt B}+ \frac{\sigma}{B}-1=0\,. \label{maybe4}
\end{align}
By \eqref{maybe1} and \eqref{maybe4},
\begin{equation}
\label{ltbb}
\partial_t B=2(\sigma-B) \quad \Rightarrow \quad 
B_t = e^{-2t}B_0+ \sigma (1-\mathrm{e}^{-2t}) \stackrel{t\rightarrow\infty}{\longrightarrow} \sigma \,.
\end{equation}
Substituting the above into \eqref{maybe2}, one gets $\partial_x A=0$. Therefore, $A=A(t)$ doesn't depend on $x$, which implies that $M(t,x)=A(t)$. With this observation, \eqref{maybe3} coincides with \eqref{ODEm1} and therefore $A(t) = M_1(t)$. In summary, densities of the form \eqref{ftpluginnl} can be a solution of the non-linear equation \eqref{nl} only if they are Gaussians of the form $ \mathcal{N}(M_1(t), B_t)$ with $B_t$ as in \eqref{ltbb} (which coincides with the evolution system $\nu_t$ if $B_0 = \sigma$).
\end{remark}

\begin{proof}[Proof of Lemma \ref{lemma:ODEforcn}]
The proof is by induction on $n$. The inductive basis, i.e., the case $n=3$, can be done by direct calculation. Assuming \eqref{eqn:dtCn=-Cn} holds for every cumulant up to order $n-1$, we want to show that the statement holds for $n$. By \eqref{eqn:commomen}, for any $n>3$,
\begin{align*}
& \dot C_n  = \dot M_n - \sum_{j=1}^{n-1} \binom{n-1}{j} \big(\dot M_j C_{n-j} + M_j \dot C_{n-j} \big)\\
&\,\,\,  = \dot M_n -  G(M_1) M_{n-1} - 2 \sigma (n-1) M_{n-2} - \sum_{j=1}^{n-1} \binom{n-1}{j} \big(\dot M_j - (n-j) M_j \big) C_{n-j}\,, 
\end{align*}
where we used the inductive assumption for $j\le n-3$ and \eqref{ODEc1}, \eqref{ODEc2} for $j=n-2, n-3$. Then, in view of \eqref{eqn:recursionmoments}, 
\begin{align*}
\dot C_n & = (n-1) G(M_1) M_{n-1}-n M_n + \sigma (n-1)(n-2) M_{n-2} \\
& \quad - \sum_{j=1}^{n-1} \binom{n-1}{j} \big(j G(M_1) M_{j-1} +\sigma j(j-1) M_{j-2} - nM_j \big) C_{n-j}  \\
& =  \, -n \bigg[ M_n  - \sum_{j=1}^{n-1} \binom{n-1}{j} M_j C_{n-j} \bigg] \\ & \quad + (n-1) G(M_1) \bigg[M_{n-1} -  \sum_{j=1}^{n-1} \binom{n-2}{j-1} M_{j-1}C_{n-j} \bigg] \\ 
& \quad + \sigma (n-1)(n-2) \bigg[ M_{n-2} - \sum_{j=2}^{n-1} \binom{n-3}{j-2} M_{j-2}C_{n-j}\bigg]\,,
\end{align*}
having used the identities, $j \binom{n-1}{j}  = (n-1) \binom{n-2}{j-1}$ and $j (j-1)\binom{n-1}{j} = (n-1)(n-2) \binom{n-3}{j-2}$.

It remains to notice that, again by \eqref{eqn:commomen}, the expression inside the first square brackets on the right-hand side is equal to $C_n$, while those inside the last two square brackets vanish. The lemma is thus proved. 
\end{proof}

\begin{remark}
\label{rem:2.3}
We are using the properties of $G$ only to prove convergence of the first moments. The behaviour of  the higher cumulants is independent  of the choice of $G$. 
\end{remark}

\subsection{Liapunov function and its rate of decay}
\label{sec:2.2}

If $f$ is a probability density with finite variance, we let
\begin{equation}
\label{Liapunov}
S(f) := \int_{\mathbb{R}}\!\mathrm{d} v \left[f(v) \log f(v) + \frac{v^2}{2\sigma} f(v) \right] + \frac{1}{\sigma} V(\langle w \rangle_f)\,,
\end{equation}
where $V(u)$ is the opposite of an antiderivative of the function $G(u)$, i.e., $V'(u)=-G(u)$.  We claim that the functional  $S(\cdot)$ is a Liapunov functional for the dynamics \eqref{p1}.  That is,  $S(\cdot)$ is bounded below and  $\frac{\mathrm{d} S(f_t)}{\mathrm{d} t} \leq 0$. Let us start by proving that the functional is   bounded from below; to this end, observe that the following inequality holds:  
\[
S(f) \ge \int_{\mathbb{R}}\!\mathrm{d} v \left[f(v) \log f(v) + \frac{v^2}{4\sigma} f(v) \right] + \frac 1{4\sigma L^2}\langle w \rangle_f^2 + \frac 1{\sigma L} V(\langle w \rangle_f)\, .
\]
Since $G$ is bounded, $V(u)/u^2 \to 0$ as $|u|\to\infty$, whence the sum of the last two terms in the right-hand side is bounded from below. As for the sum of the first two terms,  this is bounded below as well. Indeed, let 
$\rho(v)$ be the probability  measure $\rho(v)=Z (1+v^2)^{-1}$ (where $Z$ is a normalization constant). Then
$$
\int_{\mathbb{R}}\!\mathrm{d} v \left[f(v) \log f(v) + \frac{v^2}{4\sigma} f(v) \right]  = 
\int_{\mathbb{R}}\!\mathrm{d} v \left[f \log\left( f / \rho\right) + f \log \rho + f v^2\right]. 
$$
The first addend on the right-hand side is non-negative (by Pinsker's inequality). The sum of the last two addends is positive as well, as $v^2+ \log \rho = v^2 - \log (1+v^2) \geq 0 $ for every $v \in \mathbb{R}$.  

Let us now come to  compute the time derivative of $S$: 
\begin{eqnarray}
\frac{\mathrm{d}}{\mathrm{d} t} S(f_t) & = & \int_{\mathbb{R}}\!\mathrm{d} v \left[1+\frac{v^2}{2\sigma} + \log f_t(v)\right] \partial_t f_t(v) - \frac 1 \sigma  \, G(\langle w \rangle_{f_t}) \,\frac{\mathrm{d}}{\mathrm{d} t} \langle w\rangle_{f_t} \nonumber\\ 
& = &  \int_{\mathbb{R}}\!\mathrm{d} v\, \frac{1}{\sigma f_t(v)} \Big[vf_t(v)+\sigma\partial_v f_t(v)\Big] \Big[G(\langle w \rangle_{f_t})f_t(v) - vf_t(v) -\sigma\partial_v f_t(v)\Big] \nonumber\\ 
&&  ~~ -  \frac 1 \sigma G(\langle w\rangle_{f_t})  \Big[ G(\langle w\rangle_{f_t}) - \langle w\rangle_{f_t}\Big] \nonumber \\ & = &  - \int_{\mathbb{R}}\!\mathrm{d} v\, \frac{1}{\sigma f_t(v)}\Big[G(\langle w \rangle_{f_t})f_t(v) - vf_t(v) -\sigma\partial_v f_t(v)\Big]^2 \le 0\,. \label{derls}
\end{eqnarray}

\begin{remark}\label{rem:critpointsS}\textup{ Analogously to the previous computations, one can formally  show  that the Gaussians $\mu_0(v)$ and $\mu_\pm(v)$ are the unique critical points of the functional $S$ under the constraint $\int_{\mathbb{R}}\!\mathrm{d} v\, f(v) = 1$; so, when $t \rightarrow\infty$, $S(f_t)$ can only converge towards $S(f_{\infty})$, with $f_{\infty}$ being one of such Gaussians.   We already know that  $f_t$ decays towards $\mu_+$ ($\mu_-, \mu_0$, respectively) when the initial datum $f_0$ is such that $M_1(0)>0$ ($M_1(0)<0, M_1(0)=0$, respectively). In what follows we will only focus on studying the case in which $\langle w\rangle_{f_0}=M_1(0)\neq 0$. This is because if $M_1(0)=0$ then $M_1(t)=0$ for every $t\geq 0$ (from \eqref{ODEm1}); because $G(0)=0$, in this case the process reduces to the simple Orstein-Uhlenbeck process and the entropy functional simplifies to the classic form
$$
S(f) := \int_{\mathbb{R}}\!\mathrm{d} v \left[f(v) \log f(v) + \frac{v^2}{2\sigma} f(v) \right], 
$$ 
which is well studied in the literature, see \cite{CMV,Toscani} and \cite{Bakry}.  Furthermore, observe that  the equilibrium $\mu_0$ is somewhat ``unphysical": had we chosen $G$ to be as in \eqref{linG}, $\mu_0$ would not be  an invariant measure at all.
}
\end{remark}

Let us now come  to study the \textit{rate} of decay of the functional $S$. To this end, we use the by now well established Bakry-Carrillo-MacCann-Toscani-Villani strategy \cite{CMV,Toscani}. We first give an outline of how to adapt this approach to our context and then state and prove the main result of this section,  Proposition \ref{Propexpdec} below. 
To explain how we modify the Bakry-Carrillo-MacCann-Toscani-Villani strategy \cite{CMV,Toscani}, set 
$$
D_S(f_t) := - \frac{\mathrm{d}}{\mathrm{d} t} S(f_t) \,.
$$
The functional $D_S(f_t)$ is often referred to as the \textit{entropy production functional}. 
By \eqref{derls} we then  have
\begin{equation}
\label{defD}
D_S(f_t) = \int_{\mathbb{R}}\!\mathrm{d} v\,  \frac{1}{\sigma f_t}  \big[ \sigma \partial_v f_t + (v-G(M_1(t))) f_t  \big]^2\,. 
\end{equation}
Suppose there exist constants $c,\gamma,K >0$ (possibly depending on $f_0$) such that 
\begin{equation}
\label{entrprodgronw}
\frac{\mathrm{d}}{\mathrm{d} t} D_S(f_t) \le - c D_S(f_t) + K \mathrm{e}^{-\gamma t} \,.
\end{equation}
Then, integrating the above inequality on the interval $(t,s)$, one gets,
\[
D_S(f_s) -D_S(f_t) \le -c \int_t^s\!\mathrm{d} u\, D_S(f_u) + K \gamma^{-1} \big[ \mathrm{e}^{-\gamma t} -\mathrm{e}^{-\gamma s}\big] \,.  
\]
When $s$ tends to infinity, $D_S(f_s)$ tends to zero (this can be deduced from \eqref{entrprodgronw}, with calculations analogous to those leading to \eqref{expconv} below); so that, letting $s\to \infty$ in the above, one gets
\[
\begin{split}
-D_S(f_t) -  K\gamma^{-1}  \mathrm{e}^{-\gamma t} & \le - c \int_t^\infty\! \mathrm{d} u\,  D_S(f_u) = c \int_t^\infty\!\mathrm{d} u\, \frac{\mathrm{d}}{\mathrm{d} u} S(f_u) \\ & = c  [S(f_\infty) - S(f_t)] = -c  S(f_t | f_{\infty})\,,
\end{split}
\]
having set $S(f \vert f_\infty) := S(f) -S(f_\infty)$, where $f_\infty$ is intended to be equal to $\mu_{\pm}$ when $M_1(0)\gtrless 0$, respectively (see Remark \ref{rem:critpointsS}). The above then gives,
\[
\frac{\mathrm{d}}{\mathrm{d} t} S(f_t | f_\infty) = \frac{\mathrm{d}}{\mathrm{d} t} S(f_t)  = -D_S(f_t) \le -c S(f_t \vert f_{\infty}) + K\gamma^{-1}  \mathrm{e}^{-\gamma t}\,,
\]
hence exponential convergence follows, 
\begin{equation}
\label{expconv}
S(f_t | f_\infty) \le \mathrm{e}^{-ct} S(f_0 | f_{\infty}) + \begin{cases} K\gamma^{-1} |c-\gamma|^{-1}\, \mathrm{e}^{-(c\wedge\gamma) t} & \text{if } \gamma \ne c\,, \\ K c^{-1} t \, \mathrm{e}^{-ct} & \text{if } \gamma =c\,. \end{cases} 
\end{equation}
It is therefore clear that, in order to prove exponential decay of the relative entropy $S(f_t \vert f_{\infty})$, we only need to show the inequality \eqref{entrprodgronw}.  This is the purpose of the following proposition. As a premise, notice that under Hypothesis \ref{asslongtime} we have $G'(1)<1$ (and, because $G$ is an odd function, $G'$ is an even function, hence $G'(-1)=G'(1)$). 

\begin{proposition}
\label{Propexpdec}
Let Hypothesis \ref{standingassumptions} and Hypothesis \ref{asslongtime} hold. Then, for any $\gamma < 2-2G'(1)$, the inequality \eqref{entrprodgronw} holds with $c=2$ and a suitable $K>0$ depending on $G$, $M_1(0)$, and $\gamma$. Therefore, by \eqref{expconv}, if $M_1(0)\gtrless 0$ then the relative entropy $S(f_t \vert f_{\infty})$ decays exponentially fast with rate $2\wedge \gamma$.
\end{proposition}

\begin{proof}
Recalling \eqref{defD}, we have,	
\[
\begin{split}
D_S(f_t) & =  \int_{\mathbb{R}}\!\mathrm{d} v\, \frac{1}{\sigma f_t}  \big[ \sigma \partial_v f_t + (v-G(M_1)) f_t \big]^2  \\ & = \int_{\mathbb{R}}\!\mathrm{d} v\, \bigg[\frac{\sigma (\partial_v f_t)^2}{f_t} + \frac{v^2}\sigma f_t + \frac{G(M_1)^2}\sigma f_t - \frac{2G(M_1)}\sigma vf_t + 2 (v-G(M_1)) \partial_v f_t \bigg] \\ & = \int_{\mathbb{R}}\!\mathrm{d} v\, \frac{\sigma (\partial_v f_t)^2}{f_t}   + \frac1\sigma M_2+ \frac1\sigma G(M_1)^2 - \frac 2\sigma M_1G(M_1) - 2\,,
\end{split}
\]
having used the identity $ \int_{\mathbb{R}}\!\mathrm{d} v\,  2 (v-G(M_1)) \partial_v f_t = -2$ (obtained by integrations by parts). Taking the time derivative along the solutions $f_t$, and using the equations \eqref{p1}, \eqref{ODEm1}, and \eqref{ODEm2}, one gets
\[
\begin{split}
\frac{\mathrm{d}}{\mathrm{d} t} D_S(f_t) & = \int\!  \mathrm{d} v\, \bigg[\frac{2\sigma\partial_v f_t \partial_v\partial_t f_t}{f_t} - \frac{\sigma(\partial_v f_t)^2 \partial_t f_t}{f_t^2} \bigg] + \frac2\sigma \big[G(M_1)M_1 -M_2 + \sigma\big] \\ & \;\; + \frac2\sigma G(M_1)G'(M_1)(G(M_1)-M_1) - \frac 2\sigma G(M_1)(G(M_1)-M_1) \\ & \quad  - \frac2\sigma M_1G'(M_1) (G(M_1)-M_1) \\ & = \int\!  \mathrm{d} v\, \bigg[\frac{2\sigma\partial_v f_t \partial_v\partial_t f_t}{f_t} - \frac{\sigma(\partial_v f_t)^2 \partial_t f_t}{f_t^2} \bigg] + 2 - \frac2\sigma M_2 - \frac2\sigma G(M_1)^2 \\ & \quad + \frac 4\sigma M_1G(M_1) + \frac 2\sigma G'(M_1) (G(M_1)-M_1)^2\,.
\end{split}
\]
By the  definition of $D_S(f_t)$, the last identity can be rewritten in the following form, 
\begin{equation}
\label{dD=}
\frac{\mathrm{d}}{\mathrm{d} t} D_S(f_t) = -2D_S(f_t) + \frac 2\sigma G'(M_1) (G(M_1)-M_1)^2 + R\,,
\end{equation}
where
\[ 
R= \int\! \bigg[\frac{2\sigma\partial_v f_t \partial_v\partial_t f_t}{f_t} - \frac{\sigma(\partial_v f_t)^2 \partial_t f_t}{f_t^2} + \frac{2\sigma(\partial_vf_t)^2}{f_t} \bigg]\, \mathrm{d} v -2 \,.
\]
We claim $R\le 0$. To prove the claim, we integrate by parts the first term and then express $\partial_tf$ via the right-hand side of \eqref{p1}. After some straightforward computation we get,
\[
\begin{split}
R & = \int_{\mathbb{R}}\!\mathrm{d} v\, \bigg[\bigg(\frac{\sigma(\partial_v f_t)^2}{f_t^2} - \frac{2\sigma\partial_{vv}f_t}{f_t}\bigg)\partial_t f_t + \frac{2\sigma(\partial_vf_t)^2}{f_t} \bigg] -2 \\ & = \int_{\mathbb{R}}\!\mathrm{d} v\, \bigg[\frac{\sigma^2(\partial_vf_t)^2 \partial_{vv}f_t}{f_t^2} + \frac{4\sigma(\partial_vf_t)^2}{f_t} - \frac{2\sigma^2(\partial_{vv}f_t)^2}{f_t} -2 f_t\bigg]\,.
\end{split}
\]
We next observe that an integration by parts gives,
\[
\int_{\mathbb{R}}\!\mathrm{d} v\, \frac{\sigma^2(\partial_vf_t)^2 \partial_{vv}f_t}{f_t^2} = \frac 13 \int\!  \mathrm{d} v\, \frac{\sigma^2\partial_v[(\partial_vf_t)^3]}{f_t^2}\, = \frac 23 \int_{\mathbb{R}}\!\mathrm{d} v\, \frac{\sigma^2(\partial_vf_t)^4}{f_t^2}\,,
\]
which implies,
\[
\int_{\mathbb{R}}\!\mathrm{d} v\, \frac{\sigma^2(\partial_vf_t)^2 \partial_{vv}f_t}{f_t^2}  = \int\!  \mathrm{d} v\, \frac{4\sigma^2(\partial_vf_t)^2 \partial_{vv}f_t}{f_t^2}  - \int_{\mathbb{R}}\!\mathrm{d} v\, \frac{2\sigma^2(\partial_vf_t)^4}{f_t^2}\,.
\]
Using above identity to rewrite the expression of $R$ we finally obtain,
\[
\begin{split}
R & = \int\!  \mathrm{d} v\,\bigg[\frac{4\sigma^2(\partial_vf_t)^2 \partial_{vv}f_t}{f_t^2} -\frac{2\sigma^2(\partial_vf_t)^4}{f_t^2} + \frac{4\sigma(\partial_vf_t)^2}{f_t} - \frac{2\sigma^2(\partial_{vv}f_t)^2}{f_t} -2 f_t\bigg] \\ 
& = -\sigma^2 \int\!  \mathrm{d} v\, \frac {2}{f_t} \bigg[\frac{(\partial_v f_t)^2}{f_t} - \partial_{vv} f_t - \frac {f_t}{\sigma}\bigg]^2 \le 0\,.
\end{split}
\]
By \eqref{dD=} we then conclude 
\begin{equation}
\label{dD<}
\frac{\mathrm{d}}{\mathrm{d} t} D_S(f_t) \le -2D_S(f_t) + \frac 2\sigma G'(M_1) (G(M_1)-M_1)^2 \,.
\end{equation}
We finally observe that, because $G'(1)<1$, the solution $M_1(t)$ to Eq.~\eqref{ODEm1} with initial condition $M_1(0)\gtrless0 $ converges monotonically to $\pm 1$ exponentially fast, with any rate $0<\gamma'<1-G'(1)$ (recall $G'(1)=G'(-1)$ as $G$ is an odd function). In particular, given $\gamma< 2-2G'(1)$ as in the statement of the proposition and choosing $\gamma'=\gamma/2$ we deduce that
\[
\frac 2\sigma G'(M_1) (G(M_1)-M_1)^2 \le \bigg[ \frac 2\sigma\max_{|u|\le \left\vert \eta \right\vert} |G'(u)|(1+|G'(u)|)^2\bigg] (|M_1|-1)^2 \le K \mathrm{e}^{-\gamma t}\,,
\]
with $\eta:= \max\{|M_1(0)|, 1\}$ and a suitable $K>0$ depending on $\gamma$ and $M_1(0)$.\footnote{To be more precise, suppose, e.g., $M_1(0)>0$. For each fixed $0<\gamma'<1-G'(1)$ let $I$ be a neighborhood of $1$ such that $G'(\xi) -1 < - \gamma'$ for any $\xi\in I$. Therefore, $(M_1(t) -1)^2 \le (M(t_0)-1)^2 \mathrm{e}^{-2\gamma' (t-t_0)} $ for $t$ bigger than some appropriate $t_0$ (chosen such that $M(t_0)\in I$). Because $(M_1(t)-1)^2 \leq (M_1(0)-1)^2 $ for every $t\ge 0$, overall one has $(M_1(t)-1)^2 \leq (M_1(0)-1)^2 \mathrm{e}^{\gamma t_0} \mathrm{e}^{-\gamma t}$ (after setting $\gamma=2\gamma'$).} By \eqref{dD<} and the above estimate, the inequality  \eqref{entrprodgronw} follows. The proposition is thus proved.
\end{proof}

\section{Well posedness of the non-homogeneous equation}
\label{sec:3}

The main result of this section is Theorem  \ref{mainthexun}, which is an existence and uniqueness result for the nonlinear problem \eqref{nl}. The proof of  Theorem \ref{mainthexun} follows the strategy adopted in \cite{Degond} (in Remark \ref{rembeg}, point (iv),   we highlight the differences between our work and \cite{Degond}). The solution of Eq.~\eqref{nl} is constructed through a classical iterative procedure,  as a limit of  a sequence of functions, $\{f^n\}_{n\in \mathbb{N}}$; each one of the functions $f^n= f^n_t(x,v)$ solves a {\em linear} equation of kinetic Fokker-Plank type, see \eqref{genlineqn} and \eqref{lineqniter}.   This section is therefore organized as follows: after introducing the necessary notation, we start by stating several preliminary results on the linear equations satisfied by the functions $f^n$ (see Subsection \ref{sec:3.2}). The proof of such results is postponed to Appendix \ref{AppA}. In Subsection \ref{sec:3.3} we then present the proof of Theorem \ref{mainthexun}; that is, we present the iterative procedure which produces a (unique) solution in $L^1(\sqrt{1+v^2})$ for the non-linear problem \eqref{nl} (space $L^1(\sqrt{1+v^2})$ defined few lines below). We assume Hypothesis \ref{standingassumptions} to hold throughout.  

\subsection{Preliminary notation}
\label{sec:3.1}

We will work with the following function spaces:
\begin{align*}
& L^1(\sqrt{1+v^2}\,\mathrm{d} x\, \mathrm{d} v) := \left\{f\colon {\mathbb{T}}_x \times \mathbb{R}_v \to \mathbb{R} \:\colon  \int_{\mathbb{T}}\!\mathrm{d} x \int_{\mathbb{R}}\! \mathrm{d} v\, \left\vert f \right\vert  \sqrt{1+v^{2}} < \infty  \right\}\,,  
\end{align*}
and, for any $m\geq 0$, 
\begin{align*} 
& L^{\infty, m}:= L^{\infty}(\sqrt{1+v^{2m}}\,\mathrm{d} x \, \mathrm{d} v)= \{f: \mathbb{T}_x \times \mathbb{R}_v \rightarrow \mathbb{R} \;\colon \sup_{x,v}\left\vert f \right\vert \sqrt{1+v^{2m}}< \infty\}\, ,\\
& L^{2,m}  :=  L^2(\sqrt{1+v^{2m}}\,\mathrm{d} x\, \mathrm{d} v) \\ & \quad  := \left\{f\colon {\mathbb{T}}_x \times \mathbb{R}_v \to \mathbb{R} \: \colon  \int_{\mathbb{T}}\!\mathrm{d} x \int_{\mathbb{R}}\! \mathrm{d} v\, f^2 \sqrt{1+v^{2m}} < \infty  \right\}\,,  \\ 
& \mathcal{H}^{1, m}_v  := {H}^{1}(\sqrt{1+v^{2m}}\,\mathrm{d} v) := \left\{f \colon \mathbb{R}_v \to \mathbb{R} \; \colon \int_{\mathbb{R}}\! \mathrm{d} v\, (f^2+ \left\vert \partial_v f \right\vert^2 )\sqrt{1+v^{2m}} < \infty  \right\}\,, \\ 
& X^m  :=  L^2 ([0,T] \times {\mathbb{T}}_x; \mathcal{H}^{1, m}_v)\,, \;\mbox{with}
\end{align*}
\[
\| f \|_{X^m}^2:= \int_0^T\! \mathrm{d} t \!\int_{\mathbb{T}}\!  \mathrm{d} x \!\int_{\mathbb{R}}\! \mathrm{d} v\,  (f^2+ \left\vert \partial_v f \right\vert^2 )\sqrt{1+v^{2m}}\,. 
\]
We denote by $\mathcal{H}^{-1, m}_v$ the dual space of $\mathcal{H}^{1, m}_v$, namely $\mathcal{H}^{-1, m}_v:= (\mathcal{H}^{1, m}_v)^* $, so that $(X^m)^*= L^2 ([0,T] \times {\mathbb{T}}_x; \mathcal{H}^{-1, m}_v)$.\footnote{Analogously to  flat spaces, a linear functional $F$ belongs to  $\mathcal{H}^{-1, m}_v$ if and only if  there exist $f_0, f_1 \in L^{2,m}$ such that $F(g)$  can be written in the form,
$$
F(g) =   \int_{\mathbb{R}}\! \mathrm{d} v\, f_0 \, g \sqrt{1+v^{2m}} + \int_{\mathbb{R}}\! \mathrm{d} v\, f_1 \, (\partial_v g) \sqrt{1+v^{2m}}\,, \quad g \in \mathcal{H}^{1, m}_v\,.
$$} 
As customary, $L^2:= L^{2,0}$,  $H^1_v=\mathcal{H}^{1, 0}_v$, $L^{\infty}=L^{\infty,0}$, and $X:= X^0$.  Finally, if $Z$ is a row vector, then $Z^T$ denotes its transpose. If $H(x,v)$ is a vector-valued or matrix-valued function of $x$ and $v$, then $\|H\|_{L^{\infty}}$ is the sum of the $L^{\infty}$-norms (as defined above) of the components of $H$. 
 
\subsection{Well-posedness of the linear problem}
\label{sec:3.2}
As anticipated, we will construct the solution of the problem \eqref{nl} as limit of a sequence of functions $\{f^n\}_{n\geq 0}$, each of them solving a linear problem. The linear problems solved by the functions $f^n$ are all of the form
\begin{equation}\label{genlineqn}
\partial_t g  +v \partial_x g + b(t,x,v)  \partial_v g- \partial_v(vg) +c\, g- \sigma \partial_{vv}g = U(t,x,v)\,,
\end{equation}
for the unknown $g=g_t(x,v):\mathbb{R}_+ \times \mathbb{T}_x \times \mathbb{R}_v \rightarrow \mathbb{R}$, 
where the functions $b(t,x,v)$ and $U(t,x,v)$  are given and satisfy certain conditions specified later and $c \in \mathbb{R}$ is a given constant. In this section we therefore gather several results on the well posedness and on the properties of equation \eqref{genlineqn}. The weak formulation of the linear  problem \eqref{genlineqn}  is given by 
\begin{equation}\label{wvp}
E^m(g, \phi) = L^m(\phi), \quad \phi \in \Phi\,, 
\end{equation}
where $\Phi $ is the space of $C^{\infty}$ functions with compact support in $[0,T) \times \mathbb{T}_x \times \mathbb{R}_v$, $E^m$ is the bilinear form
\[
\begin{split}
E^m(g,\phi) & =\int_0^T\!\mathrm{d} t\int_{\mathbb{T}}\!\mathrm{d} x \int_{\mathbb{R}}\!\mathrm{d} v\,  g \big( - \partial_t \phi -v \partial_x \phi + (c-1) \phi  \big) \sqrt{1+v^{2m}} \\  & \quad  + \int_0^T\!\mathrm{d} s\int_{\mathbb{T}}\!\mathrm{d} x \int_{\mathbb{R}}\!\mathrm{d} v\,   \partial_v g  \big[ (b-v) \phi + \sigma \partial_v \phi\big]\sqrt{1+v^{2m}} \\ & \quad +  \sigma \int_0^T\!\mathrm{d} s\int_{\mathbb{T}}\!\mathrm{d} x \int_{\mathbb{R}}\!\mathrm{d} v\,  (\partial_v g) \phi  \frac{mv^{2m-1}}{\sqrt{1+v^{2m}}}\,,
\end{split}
\]
and $L^m$ is the functional
\[
L^m(\phi) = \int_0^T\!\mathrm{d} s\int_{\mathbb{T}}\!\mathrm{d} x \int_{\mathbb{R}}\!\mathrm{d} v\,  U \phi \sqrt{1+v^{2m}} + \int_{\mathbb{T}}\!\mathrm{d} x \int_{\mathbb{R}}\!\mathrm{d} v\,  \sqrt{1+v^{2m}} g(0,x,v) \phi (0,x,v) \,.
\]
\begin{proposition}
\label{propwellposlin} Consider the  problem  \eqref{genlineqn}, 
where the coefficient $b(t,x,v)$ and the forcing $U(t,x,v)$ satisfy the following assumptions: 
\begin{itemize}
\item[(i)] $b(t,x, v)$ is smooth and bounded; 
\item[(ii)] there exists some  $m\geq 1$ such that $U \in (X^m)^*=L^2([0,T] \times {\mathbb{T}}_x; \mathcal{H}^{-1, m}_v )$.
\end{itemize}
Then, for any initial condition $g_0 \in L^{2,m}$,  the problem \eqref{genlineqn} admits a unique solution in  $X^m$. This solution is unique in the sense that if $\tilde{g}$ is another solution in ${X}^m$, then $\| g-\tilde{g}\|_{X^m}=0$.  Moreover, the solution belongs to the space
\begin{equation}
\label{spacesol}
\mathcal{Y}^m := \big\{g \in X^{m}=L^2([0,T] \times \mathbb{T}; \mathcal{H}^{1, m}_v) \;\colon \partial_t g + v \partial_x g \in (X^{m})^*\big\}\,.
\end{equation}
\end{proposition}

\begin{proof}
See Appendix \ref{AppA}. 
\end{proof}

\begin{remark}
\label{rembeg} 
Some remarks on Eq.~\eqref{genlineqn} and Proposition \ref{propwellposlin}. 
\begin{itemize}
\item[(i)] If the function $g_t(x,v)$ solves Eq.~\eqref{genlineqn} with $c=0$, then  $e^{-ct}g_t(x,v)$ solves Eq.~\eqref{genlineqn}. The results on the solution of Eq.~\eqref{genlineqn} that we prove below (Lemma \ref{lemmamaxprinc}, Lemma \ref{lemmaduhamel}, and Lemma \ref{lemmacriterio}), are unaltered by this change of unknown and will therefore only need to be proven for $c=0$. 
\item[(ii)]If $b$ and $U$ are smooth, all the derivatives in Eq.~\eqref{genlineqn} can be intended in classical sense, due to the hypoellipticity \cite{Hoermander} on $\mathbb{R}_+ \times \mathbb{T}_x \times \mathbb{R}_v$ of the second order differential operator $\mathcal{L}_t$ defined by
\[
\mathcal{L}_tf = \partial_tf  + v \partial_xf  +b(t,x,v) \partial_vf  - \partial_v(v f)- \sigma \partial_{vv}f\,. 
\]
\item[(iii)] Assumption ii) of Proposition \ref{propwellposlin} is satisfied as soon as 
\begin{equation}\label{aaa}
U\in L^2([0,T]; L^{2,m} ), \quad {\mbox{i.e. }} \,\,\, \int_0^T\!\mathrm{d} t\, \| U(t) \|_{L^{2,m}}^2 < \infty\,.
\end{equation}
 (As $H^1 \subset L^2 \subset H^{-1}$, and this is true for the weighted spaces at hand as well.)
\item[(iv)] This proof of Proposition \ref{propwellposlin} follows  the structure  of proof adopted in  \cite[Prop.~A.1]{Degond}. To compare with \cite{Degond}, observe that  Eq.~\eqref{genlineqn} can be rewritten as 
\[
\partial_t g  +v \partial_x g + (b(t,x,v) - v)  \partial_v g+(c-1) g - \sigma \partial_{vv}g = U(t,x,v)\,.
\]
Setting $a(t,x,v)=b(t,x,v)-v$, the above is of the form,  
\begin{equation}
\label{sf}
\partial_t g  +v \partial_x g + a(t,x,v) \partial_v g+(c-1) g - \sigma \partial_{vv}g = U(t,x,v)\,.
\end{equation}
Equations of this type are analized in \cite{Degond}, see \cite[Eqs.~(40) and (47)]{Degond}, where well posedness of the above is shown under the assumption that $a(t,x,v)$ is bounded and with bounded $v$-divergence.  Clearly, the function $a(t,x,v)=b(t,x,v)-v$ is not bounded,  so we can't just use \cite[Prop.~A.1]{Degond} directly.  Nonetheless, the same scheme of proof applies and, with slight modifications, it also allows to remove the assumption that $a$ should have bounded $v$-divergence (which, in our case, translates into dropping the assumption $\|\partial_vb\|_{\infty}< \infty$).  It is important to notice that while in \cite{Degond} the author works in flat $L^2$ spaces, we work in the weighted spaces introduced in Subsection \ref{sec:3.1}.  This is done in order to deal with the fact that the coefficient $a$ is unbounded and to obtain information about  the integrability (in the velocity variable) of the solution of \eqref{sf}.  
\end{itemize} 
\end{remark}

\begin{lemma}[Maximum principle]
\label{lemmamaxprinc} 
Consider the linear equation \eqref{genlineqn} and let the assumptions of Proposition \ref{propwellposlin} hold. If the initial datum $g_0(x,v)$ has moments of order at least $2$, i.e., $g_0 \in L^{2,m}$ with $m\geq 2$,  the following holds. 
\begin{itemize}
\item[(a)]  If $g_0(x,v) \geq 0$ and $U\geq 0$ then  $g_t(x,v) \ge 0$ for all $t\ge 0$.
\item[(b)]  As a consequence of the above,   
\begin{equation}
\label{maxprineqn}
\| g_t \|_{L^\infty} \leq \| g_0\|_{L^\infty} + \int_0^t\! \mathrm{d} s\, \| U(s)\|_{L^\infty}\,.
\end{equation}
Therefore, if $g_0 \in  L^{\infty}(\mathbb{T}_x \times \mathbb{R}_v)$ and $U \in L^1([0,T]; L^{\infty}(\mathbb{T}_x \times \mathbb{R}_v))$  then $g_t \in L^{\infty}(\mathbb{T}_x \times \mathbb{R}_v) $ for all $t\ge 0$. 
\end{itemize}

\end{lemma}
\begin{proof}
See Appendix \ref{AppA}. 
\end{proof}

\begin{lemma}[Duhamel Formula]
\label{lemmaduhamel} 
Consider the linear equation \eqref{genlineqn} and let  the assumptions of Pro\-position 
\ref{propwellposlin} hold.   Assume additionally that $U \in L^1([0,T] \times \mathbb{T}_x \times \mathbb{R}_v)$, $\partial_v b \leq 0$ and the initial datum belongs to $L^1(\mathbb{T}_x \times \mathbb{R}_v) \cap L^{2,m}$, for some $m\geq 2$;   then the following Duhamel formula holds,
\begin{equation}
\label{Duhamel}
\| g_t\|_{L^1} \leq \|g_0\|_{L^1} + \int_0^t\! \mathrm{d} s\, \|U(s)\|_{L^1}\,. 
\end{equation}
\end{lemma}

\begin{proof}
See Appendix \ref{AppA}. 
\end{proof}

The practical (and straightforward) criterion that we will use in order to make sure that all of the above results hold is the following lemma, which puts an emphasis on the assumptions on $U$ that one needs to verify in practice (typically the assumptions on $b$ will be straightforward to check). 
\begin{lemma}
\label{lemmacriterio}
Consider the linear equation \eqref{genlineqn}. Suppose the initial datum 
$g_0$ and the coefficient $b(t,x,v)$ satisfy the assumptions of Proposition \ref{propwellposlin}, Lemma \ref{lemmamaxprinc} and Lemma \ref{lemmaduhamel}. If the forcing term $U$  satisfies \eqref{aaa} and  
\begin{equation}
\label{bbb}
U \in L^1([0,T];  L^{\infty}((1+v^2)^{\ell/2} )\quad \mbox{for some } \ell>1\,,
\end{equation}
then all the results of Proposition \ref{propwellposlin}, Lemma \ref{lemmamaxprinc} and Lemma \ref{lemmaduhamel} do hold.
\end{lemma}

\begin{proof}
We only need to check that our assumptions on $U$ imply the assumptions on $U$ made in Proposition \ref{propwellposlin}, Lemma \ref{lemmamaxprinc} and Lemma \ref{lemmaduhamel}. Condition \eqref{aaa} implies ii) of Proposition \ref{propwellposlin} (see item iii) of Remark \ref{rembeg}). Moreover,  if \eqref{bbb} holds then  $U$ is clearly in $L^{\infty}$ and therefore also in $L^{1}$ in the $x$-variable. As for the behaviour in the $v$-variable, we recall the following interpolation inequality (see \cite[Lemma B.1]{Degond}): for any $\ell>1$ there exists $C>0$ such that,  for any function $u=u(v)\colon \mathbb{R} \rightarrow \mathbb{R}$,
\begin{equation}
\label{interp}
\| u\|_{L^1} \leq C \| u\|_{L^\infty}^{(\ell-1)/\ell} \| (1+v^2)^{\ell/2} u \|_{L^\infty}^{1/\ell}\,.
\end{equation}
From this, the statement follows. 
\end{proof}

\subsection{Construction of the solution to the nonlinear problem: the iterative procedure}
\label{sec:3.3}

After the preliminary results of Subsection \ref{sec:3.2}, this section is devoted to the proof of the main existence and uniqueness result, Theorem \ref{mainthexun} below. Before stating it, we clarify that the weak formulation of problem \eqref{nl}  is given by
\[
E_{NL}(f, \phi) = \int_{\mathbb{T}}\!\mathrm{d} x \int_{\mathbb{R}}\!\mathrm{d} v\,  f(0,x,v) \phi (0,x,v)\,, \qquad \phi \in \Phi \,,
\]
where we recall that $\Phi $ is the space of $C^{\infty}$ functions with compact support in $[0,T) \times \mathbb{T}_x \times \mathbb{R}_v$ and the functional $E_{NL}$ is defined as
\begin{equation}
\label{nonlintermwf}
\begin{split}
E_{NL}(f,\phi) & = \int_0^T\!\mathrm{d} t\int_{\mathbb{T}}\!\mathrm{d} x \int_{\mathbb{R}}\!\mathrm{d} v\,  f \big( - \partial_t \phi -v \partial_x \phi - \phi  \big)   \\ & \quad + \int_0^T\!\mathrm{d} s\int_{\mathbb{T}}\!\mathrm{d} x \int_{\mathbb{R}}\!\mathrm{d} v\,   \partial_v f  \big[ -v \phi + \sigma \partial_v \phi\big] \\
& \quad + \int_0^T\!\mathrm{d} s\int_{\mathbb{T}}\!\mathrm{d} x \int_{\mathbb{R}}\!\mathrm{d} v f \, G \left( M(s,x)\right) \partial_v \phi\,, 
\end{split}
\end{equation}
with $M(t,x)$ defined in \eqref{Mnl}. In what follows, we set  $D f(x,v) = (\partial_x f,\partial_v f)$ and $D^2f(x,v) = \begin{pmatrix} \partial_{xx}f & \partial_{xv}f \\ \partial_{xv} f & \partial_{vv} f \end{pmatrix}$.

\begin{theorem}
\label{mainthexun}
Consider the nonlinear equation \eqref{nl} and let Hypothesis \ref{standingassumptions} hold.  Then, for any initial datum $f_0$ satisfying the following assumptions,
\begin{itemize}
\item $f_0 \in L^{2,m}$ for some $m\geq 8$,
\item $( \left\vert D f_0\right\vert + \left\vert D^2 f_0\right\vert ) \sqrt{1+v^{6}} \in L^{\infty} \cap L^{2,m}$ for some $m \geq 2$,
\end{itemize}
there exists a unique weak solution of \eqref{nl}. The solution belongs to the space  $L^{\infty}([0,T]; L^1(\sqrt{1+v^2})\,  \mathrm{d} x\, \mathrm{d} v)$.
\end{theorem}

We first explain the strategy of proof and state two necessary technical lemmata, Lemma \ref{lemmaequilim} and Lemma \ref{equilm2}.  We then move on to the proof of Theorem \ref{mainthexun}.

We construct the (unique) solution to the nonlinear Eq.~\eqref{nl} with initial condition $f_0(x,v)$, as the limit of a sequence of functions $\{ f^n(x,v)\}_{n\geq 0}$ which are recursively defined according to the following scheme: we set $f^0_t(x,v) := f_0(x,v)$ and, for any $n \ge 1$, we let $f^n_t(x,v)$ be the solution to the {\em linear} equation,
\begin{equation}
\label{lineqniter}
\begin{cases}
\partial_t f^n_t = -v \partial_x f^n_t - G(M^{n-1}(t,x)) \partial_v f^n_t + \partial_v(vf^n_t)+ \sigma \partial_{vv} f^n_t \,, \\ f^n_0(x,v)=f_0(x,v)\,, \end{cases}
\end{equation}
 where (cfr \eqref{Mnl})
\begin{equation}
	\label{Mn-1}
M^{n-1}(t,x) := \frac{\langle w \rangle_{f^{n-1}_t,\varphi}}{\langle 1 \rangle_{f^{n-1}_t,\varphi}} =  \frac{1}{\rho^{n-1}(t,x)} \int_{\mathbb{T}}\!\mathrm{d} y  \int_{\mathbb{R}}\!\mathrm{d} w\,  f^{n-1}_t(y,w)\, \varphi(x-y) \, w\,,
\end{equation}
with
\[
\rho^{n-1}(t,x) := \langle 1 \rangle_{f^{n-1}_t,\varphi} = \int_{\mathbb{T}}\! \mathrm{d} y \int_{\mathbb{R}}\!\mathrm{d} w\, f^{n-1}_t(y,w) \,  \varphi(x-y)\,.
\]
\begin{remark}\label{remfn}\rm
Existence and uniqueness of Eq.~\eqref{lineqniter} follows from Proposition \ref{propwellposlin} with $U=0, c=0$ and $b(t,x,v)=G(M^{n-1}(t,x))$. Because $G$ is assumed to be bounded,  Proposition \ref{propwellposlin} can be applied to Eq.~\eqref{lineqniter}, which is therefore  well-posed in the space $\mathcal{Y}^m$ for any initial datum in $L^{2,m}$.  Moreover, the following fundamental lemmata hold. 
\end{remark}

\begin{lemma}[Estimates on the solution of  Eq.~\eqref{lineqniter}]  
\label{lemmaequilim}
Let $f^n$ be the solution of Eq.~\eqref{lineqniter} with an initial condition $f_0$; assume the interaction function $\varphi$ appearing in \eqref{Mn-1} is smooth and strictly positive.  Set, for $m\geq 1$,
\begin{equation}
\begin{split}
\label{yz}
& Y^{n,m}_t(x,v) := f^n_t(x,v) \sqrt{1+v^{2m}}\,, \\ & Z^{n,m}_t(x,v) := (D f^n_t(x,v))^\top \sqrt{1+v^{2m}}\,, \\ & H^{n,m}_t(x,v) := D^2f^n_t(x,v) \sqrt{1+v^{2m}}\,.
\end{split}
\end{equation} 
Then the following holds. 
\begin{itemize}
\item[(a)]  If $Y_0^{0,m}=f_0 \sqrt{1+v^{2m}} \in L^{2,2} \cap L^{\infty}$,  then there exists a non-negative function $\alpha_m(t)$ which is bounded on compacts and such that 
\begin{equation}
\label{fest}
\| Y^{n,m}_t\|_{L^\infty} \le \alpha_m(t)\qquad \forall\, n\ge 0\,. 
\end{equation}
\item[(b)] If $f_0 \in L^{2,2m+1}$, $\sqrt{1+v^{2m}} D f_0 \in L^{2,2} \cap L^{\infty}$,  then there exists a non-negative function $\beta_m(t)$ which is  bounded on compacts and such that 
\begin{equation}
\label{sest}
\| Z^{n,m}_t\|_{L^\infty} \le \beta_m(t) \qquad \forall\, n\ge 0\,. 
\end{equation}
\item[(c)] If $(f_0 + |D f_0| + |D^2f_0|) \sqrt{1+v^{2m}} \in L^{2,2} \cap L^{\infty}$,  then there exists a non-negative function $\gamma_m(t)$ which is  bounded on compacts and such that 
\begin{equation}
\label{dest}
\|H^{n,m}_t\|_{L^\infty} \le \gamma_m(t) \qquad \forall\, n\ge 0\,. 
\end{equation}
\end{itemize}
\end{lemma}

\begin{proof} 
See Appendix \ref{AppA}.
\end{proof}

\begin{remark}\textup{ If $f_0 \in L^{2,5}$ and $ \sqrt{1+v^4}\, D f_0 \in L^{2,2} \cap L^{\infty}$,  using  point (b) of Lemma \ref{lemmaequilim}, the interpolation inequality \eqref{interp} and the continuity in $x$  of the integral function on the left-hand side (which is true by hypoellipticity), one has 
\begin{equation}
\label{suptir}
\sup_{x\in \mathbb{T}}\int_{\mathbb{R}}\! \mathrm{d} v\, |\partial_v f^{n}_t(x,v) | \le \kappa(t) \qquad \forall\, n\ge 0\,,
\end{equation}
where $\kappa(t)$ denotes (here and in the following) a generic time-dependent function, bounded on compact sets. Similarly,  if $vf_0 \sqrt{1+v^4}$ (which grows at infinity like $f_0 \sqrt{1+v^6}$) belongs to  $L^{\infty} \cap L^{2,2}$ then 
\begin{equation}
\label{suptirv}
\sup_{x\in \mathbb{T}}\int_{\mathbb{R}}\! \mathrm{d} v\, |v| f^{n}_t(x,v) \le \kappa(t) \qquad \forall\, n\ge 0\,. 
\end{equation}
Finally, let $h^n$ be defined as in \eqref{hn}. Since $\partial_vh^{n}= \left( v/\sqrt{1+v^2}\right) f^{n}+ \sqrt{1+ v^2} \partial_v f^{n}$, the same kind of reasoning shows that if $f_0 \in L^{2,7}$ and $\sqrt{1+v^6}\, D f_0 \in L^{2,2} \cap L^{\infty}$, then
\begin{equation}\label{thirdsup}
\sup_{x \in \mathbb{T}} \int\! \mathrm{d} v\, \big(|\partial_vh^{n}_t| +  |f^{n}_t|\big)  \leq \kappa (t) \qquad \forall\, n\ge 0\,. 
\end{equation}
}
\end{remark}

\begin{lemma} 
\label{equilm2}
Let $f^n$ be the solution of Eq.~\eqref{lineqniter}. With the notation introduced so far, if $f_0$ and $\partial_v f_0$ both belong to $L^{2,m}$ with $m\geq 3$, then there exists a constant $c>0$ such that
\begin{equation}
\label{espest}
\|f^n_t\|_{L^2} \leq  e^{ct} \|f_0\|_{L^2}\,, \quad   \int_{\mathbb{T}}\! \mathrm{d} y \int_{\mathbb{R}}\!\mathrm{d} w\,  (w f^n_t)^2 \leq  e^{ct}  \int_{\mathbb{T}}\! \mathrm{d} y \int_{\mathbb{R}}\!\mathrm{d} w\, (1+w^2) \left\vert f_0 \right\vert^2 \,.
\end{equation}
\end{lemma}
\begin{proof} 
They are consequence of the following identities, 
\[
\begin{split}
\frac{\mathrm{d}}{\mathrm{d} t} \|f^n_t\|_{L^2}^2 & =  \|f^n_t\|_{L^2}^2 - \sigma\|\partial_v f^n_t\|_{L^2}^2 \,, \\ 
\frac{\mathrm{d}}{\mathrm{d} t} \|vf^n_t\|_{L^2}^2 & = 2 \int_{\mathbb{T}}\!\mathrm{d} x\! \int_{\mathbb{R}}\!\mathrm{d} v\, G(M^{n-1}(t)) v (f^n_t)^2  \\ & \quad -  \|vf^n_t\|_{L^2}^2 + 2  \|f^n_t\|_{L^2}^2 - 2  \sigma\|v\partial_vf^n_t\|_{L^2}^2\,,
\end{split}
\]
which can be easily verified by \eqref{lineqniter} and integration by parts. In doing such integrations by parts,  notice that, by Lemma \ref{lemmaequilim}, $v^3 (f^n_t)^2$ and $v^2 f^n_t \partial_v f^n_t$ vanish at infinity. 
\end{proof}

\begin{proof}[Proof of Theorem \ref{mainthexun}]
In the sequel, we will assume $t\in [0,T]$ where $T>0$ is any fixed time, and we will denote by $C$ a generic positive constant, whose numerical value (possibly depending on $f_0$ and $T$) may change from line to line. Analogously, $\kappa(t)$ will denote a generic time-dependent function, bounded on compact sets; the specific expression of this function may change from line to line.  
With the notation introduced so far, the scheme of proof is classical: after setting
$$
\xi^n(t) := \max_{s\in [0,t]} \big(\|f^{n}_s - f^{n-1}_s\|_{L^1} + \|h^{n}_s-h^{n-1}_s\|_{L^1}\big)\,,
$$
with
\begin{equation}
\label{hn}
h^n_t(x,v) :=Y_t^{n,1}=  f^n_t(x,v) \sqrt{1+v^2}\,, 
\end{equation}
the main point (and the lengthy part) of the proof consists in showing the contraction
\begin{equation}\label{eqn:conclusion}
\xi^n(t) \le C \int_0^t\!\mathrm{d} s\, \xi^{n-1}(s)\qquad \forall\, t\in [0,T]\,.
\end{equation}
Indeed,  by iterating \eqref{eqn:conclusion} we deduce that $f^n$ and $\sqrt{1+v^2}f^n$ converge to $f$ and $\sqrt{1+v^2}f$, respectively, in $L^1(\mathbb{T}_x \times \mathbb{R}_v)$  (hence, also $vf^n \xrightarrow{L^1} vf$). To prove that the limit of the sequence $f^n$ is actually a solution of the nonlinear PDE (and in particular to pass to the limit in the nonlinear term \eqref{nonlintermwf}) one uses such a convergence plus Lemma \ref{lemmaequilim}, which implies the existence of a subsequence which converges weak* in $L^{\infty}([0,T] \times \mathbb{T}_x \times \mathbb{R}_v)$. We omit the details, which are standard. 
This proves existence of a solution of the nonlinear problem. Uniqueness can be obtained with calculations similar to those that allow one to derive \eqref{eqn:conclusion}. We therefore only need to prove \eqref{eqn:conclusion}. This is done in three steps (with calculations that are standard but lengthy, so in places we only indicate how to complete them). 

\smallskip

{$\bullet $ \em Step 1.} This step consists in showing the following inequality
\begin{equation}
\label{contrl1}
\| f^n_t-f^{n-1}_t\|_{L^1}  \le C \int_0^t\!\mathrm{d} s\, \kappa(s)  \, \big( \| h^{n-1}_s-h^{n-2}_s\|_{L^1} + \|f^{n-1}_s- f^{n-2}_s\|_{L^1} \big)\,.
\end{equation}
In order to show the above, we consider the function $g=g_t(x,v)$ defined as $g := f^n - f^{n-1}$. Such a function satisfies the {\em linear} equation,  
\begin{equation}
\label{diffeqn}
\begin{split}
& \partial_t g + v \partial_x g + G(M^{n-1}) \partial_v g - \partial_v({vg})  - \sigma\partial_{vv}g = \left[ G(M^{n-2})-G(M^{n-1})\right] \partial_v f^{n-1}\,, \\ & g_0 = 0\,, 
\end{split}
\end{equation}
which is of the form \eqref{genlineqn} with 
\[
b =  G(M^{n-1})\,, \quad c=0\,, \quad  U = \left[ G(M^{n-2})-G(M^{n-1})\right] \partial_v f^{n-1}\,.
\]
By Proposition \ref{propwellposlin} (applied to Eq.~\eqref{lineqniter}),  $\partial_vf^{(n-1)} \in L^{2,m}$ when the initial datum $f_0$ belongs to $L^{2,m}$ ($m\geq 2$). So, under this assumption on $f_0$, we have, $U \in  L^2([0,T] \times \mathbb{T}; \mathcal{H}^{-1, m}_v)$. We therefore have well posedness of Eq.~\eqref{diffeqn} in $\mathcal{Y}^m$,  if the initial data of Eq.~\eqref{lineqniter} belong to $L^{2,m}$ (for $m\geq 2$). 

Observe now that $U \in L^1([0,T] \times \mathbb{T} \times \mathbb{R})$ since $(\partial_v f_0)\sqrt{1+v^{2m}} \in L^{2,2} \cap L^{\infty}$ ($m \ge 2$).  This can be seen by using point (b) of  Lemma \ref{lemmaequilim}, together with the  interpolation inequality \eqref{interp}.
Under such an assumption on the initial datum $f_0$, the Duhamel formula of Lemma \ref{lemmaduhamel} holds. Using this fact and the Lipschitzianity of $G$,  we obtain
\[
\|f^n_t-f^{n-1}_t\|_{L^1} \le \int_0^t\!\mathrm{d} s\, \| (\partial_v f^{n-1}_s) (M^{n-1}(s)-M^{n-2}(s))\|_{L^1}\,.
\]
Looking at the integrand on the right-hand side,
\[
 \| (\partial_v f^{n-1}_t) (M^{n-1}(t)-M^{n-2}(t))\|_{L^1} = \int_{\mathbb{T}}\!\mathrm{d} x\, | M^{n-1}(t)-M^{n-2}(t) | \int_{\mathbb{R}}\!\mathrm{d} v\, |\partial_v f^{n-1}_t|  \,.
\]
The inner integral can be estimated by using \eqref{suptir}. Moreover, 
\begin{equation}
\label{mm}
\begin{split}
& |M^{n-1}(t)-M^{n-2}(t)|  = \bigg| \frac{\langle w \rangle_{f^{n-1}_t,\varphi}}{\rho^{n-1}(t)} -  \frac{\langle w \rangle_{f^{n-2}_t,\varphi}}{\rho^{n-2}(t)}\bigg| \\ & \qquad\quad \le \frac{\big|\langle w \rangle_{f^{n-1}_t,\varphi} - \langle w \rangle_{f^{n-2}_t,\varphi}\big| \rho^{n-2}(t) + \big |\langle w \rangle_{f^{n-2}_t,\varphi}\big|  \, \big|\rho^{n-1}(t)-\rho^{n-2}(t)\big|}{\rho^{n-1}(t) \rho^{n-2}(t)} \\ & \qquad\quad \le C  \big( \| h^{n-1}_t-h^{n-2}_t\|_{L^1} + \|f^{n-1}_t- f^{n-2}_t\|_{L^1} \big)\,,
\end{split}
\end{equation}
having used the fact that $\varphi$ is positive and uniformly bounded below,  $\varphi \geq \epsilon >0$,   and \eqref{suptirv}.  Therefore \eqref{contrl1} follows.

\smallskip

$\bullet$ {\em Step 2. } The second step consists in bounding the first addend on the right-hand side of \eqref{contrl1}. In particular, we will show the following inequality,
\begin{equation}
\label{boh}
\begin{split}
 \| h^n_t - h^{n-1}_t\|_{L^1} &  \leq C \int_0^t\!\mathrm{d} s\ \int_0^s\!\mathrm{d} r\, \kappa(s) \big(\|h^{n-1}_r-h^{n-2}_r\|_{L^1} +\|f^{n-1}_r-f^{n-2}_r\|_{L^1} \big) \\ 
& \quad + C  \int_0^t\!\mathrm{d} s\, \big(\|h^{n-1}_s-h^{n-2}_s\|_{L^1} + \|f^{n-1}_s-f^{n-2}_s\|_{L^1} \big) \\ & \quad + C \int_0^t\!\mathrm{d} s\, \| \partial_v(f^n_s -f^{n-1}_s) \|_{L^1}\,. 
\end{split}
\end{equation}
To this end, let  $\tilde g=\tilde g_t(x,v)$ be defined as $\tilde g:=h^n-h^{n-1}$. From  \eqref{hn} and  \eqref{Yeq} we deduce that the function  $\tilde g$ solves the equation,
\begin{equation}
\label{gtilde} 
\partial_t \tilde g + v \partial_x\tilde g + G(M^{n-1}(t)) \partial_v\tilde g - \partial_v (v\tilde g) + \tilde g - \sigma \partial_{vv} \tilde g = R^n\,, \quad \tilde g_0 = 0\,, 
\end{equation}
where, setting $A_1(v) := v/\sqrt{1+v^2}$ and recalling the definition \eqref{Yeqr},
\[
\begin{split}
R^n & = \big[G(M^{n-2}) - G(M^{n-1}) \big] \partial_v h^{n-1}+ h^n - h^{n-1} +   R^{n,1}- R^{(n-1),1} \\ & \quad -\, 2 \sigma \frac{A_1(v)}{\sqrt{1+v^2}}\partial_v\tilde g \\ & = \big[G(M^{n-2}) - G(M^{n-1}) \big] \big(\partial_v h^{n-1} - A_1(v) f^{n-1} \big)  - 2\sigma A_1(v) \partial_v(f^n-f^{n-1})\\ & \quad +  (f^n-f^{n-1})  \bigg(G(M^{n-1}) A_1(v) + \frac{1}{\sqrt{1+v^2}} - \sigma A_1'(v)  \bigg)\,.
\end{split}
\]
Because $A_1, A_1':=\partial_v A_1$ and   $G$ are  bounded, all the functions that multiply the differences $(f^n-f^{n-1})$ and $\partial_v(f^n-f^{n-1})$ on the right-hand side above are bounded. Analogous reasoning holds for the difference $R^{n,1}-R^{(n-1),1}$ (see \eqref{Yeqr}).  Moreover,  $\partial_vh^{n-1}= A_1(v) f^{n-1}+ \sqrt{1+ v^2} \partial_v f^{n-1}$, so $R^n \in L^{2,1}$ as soon as $f_0 \in L^{2,3}$, so that the well posedness of \eqref{gtilde} is ensured.  By the Lipschitzianity of $G$ we can then write
\[
|R^n| \le C \big(|\partial_vh^{n-1}_t| +  |f^{n-1}_t|\big) |M^{n-1}-M^{n-2}| + C |f^n-f^{n-1}| + C |\partial_v(f^n-f^{n-1})| \,.
\]
Now observe that the right-hand side of \eqref{gtilde} is in $L^1$ as, by our assumptions on $f_0$, it is in $L^{\infty,2}$.  We can then apply  the Duhamel formula (to \eqref{gtilde}) and use \eqref{thirdsup} to obtain
\begin{equation}
\begin{split}
 \| h^n_t  - h^{n-1}_t & \|_{L^1} \le C \int_0^t\!\mathrm{d} s\, \| f^n_s - f^{n-1}_s\|_{L^1}  + C \int_0^t\!\mathrm{d} s\, \| \partial_v(f^n_s -f^{n-1}_s) \|_{L^1} \\ &  \quad + C \int_0^t\!\mathrm{d} s\, \kappa(s) \int_{\mathbb{T}}\mathrm{d} x\, |M^{n-1}(s) - M^{n-2}(s)|  \\
&  \le C \int_0^t\!\mathrm{d} s\ \int_0^s\!\mathrm{d} r\, \kappa(r) \big(\|h^{n-1}_r-h^{n-2}_r\|_{L^1} +\|f^{n-1}_r-f^{n-2}_r\|_{L^1} \big) \\ 
&  \quad + C  \int_0^t\!\mathrm{d} s\, \big(\|h^{n-1}_s-h^{n-2}_s\|_{L^1} + \|f^{n-1}_s-f^{n-2}_s\|_{L^1} \big) \\ & \quad + C \int_0^t\!\mathrm{d} s\, \| \partial_v(f^n_s -f^{n-1}_s) \|_{L^1}\,,
\end{split}
\end{equation}
having used \eqref{mm} and \eqref{contrl1} in the last inequality. Hence \eqref{boh} is proved. 

$\bullet$ {\em Step 3.} We now need an estimate of the last term on the right-hand side of \eqref{boh}. To this end,  acting like in the proof of Lemma \ref{lemmaequilim}, we observe that the differences $p(x,v) = \partial_v (f^n - f^{n-1})$ and $q(x,v) = \partial_x (f^n-f^{n-1})$ satisfy PDEs with a structure similar to the one of the equations encountered so far (see \eqref{eqforz} and comments thereafter), so one can apply again the strategy that we have already used and that we only sketch in this case. By using \eqref{eqforz} we find an equation for $p(x,v)$. The Duhamel formula applied to such an equation gives 
\begin{equation}
\label{duhampavf}
\begin{split}
\| \partial_v f^n_t - \partial_v f^{n-1}_t \|_{L^1} & \le \int_0^t\!\mathrm{d} s\, \|(G(M^{n-1}(s) - G(M^{n-2}(s)) \partial _{vv} f^{n-1}_s  \|_{L^1} \\ & \quad + \int_0^t\!\mathrm{d} s\, \| \partial_x f^n_s - \partial_x f^{n-1}_s \|_{L^1} \,.
\end{split}
\end{equation} 
Acting similarly, one also gets
\begin{equation}
\label{duhampaxf}
\begin{split}
\| \partial_x f^n_t - \partial_x f^{n-1}_t \|_{L^1} & \le \int_0^t\!\mathrm{d} s\, \|\left[ G(M^{n-1}(s) - G(M^{n-2}(s) \right] \partial _{xv} f^{n-1}_s  \|_{L^1} \\ & \quad + \int_0^t\!\mathrm{d} s\, \| \partial_x G(M^{n-1}(s)) \partial_v f^n_s -  \partial_x G(M^{n-2}(s)) \partial_v f^{n-1}_s \|_{L^1}\,. 
\end{split}
\end{equation}

Using point (c) of Lemma \ref{lemmaequilim} and the interpolation inequality \eqref{interp} we have,
\[
\sup_{x \in \mathbb{T}} \int\! \mathrm{d} v\, \big(|\partial_{xv}f^{n-1}_t| +  |\partial_{vv} f^{n-1}_t|\big)  \le \kappa(t)\,.
\]
Therefore, the first addends on the right-hand side of  \eqref{duhampavf} and \eqref{duhampaxf} can be treated in a standard way. Concerning the last addend on the right-hand side of \eqref{duhampaxf}, from \eqref{stgx} and point (b) of Lemma \ref{lemmaequilim},
\[
\begin{split}
& \| \partial_x G(M^{n-1}) \partial_v f^n - \partial_x G(M^{n-2}) \partial_v f^{n-1} \|_{L^1} \\ &\quad \le \| (\partial_x G(M^{n-1}) - \partial_x G(M^{n-2})) \partial_v f^n \|_{L^1} + \| \partial_x G(M^{n-2})  \partial_v (f^n - f^{n-1}) \| _{L^1} \\ &\quad  \le C \|\partial_x G(M^{n-1}) - \partial_x G(M^{n-2})\|_{L^1} + C \| \partial_v (f^n - f^{n-1}) \|_{L^1}\,.
\end{split}
\]
Now, again by \eqref{stgx} and using the Lipschitzianity of $G$,
\[
|\partial_x G(M^{n-1}) - \partial_x G(M^{n-2}) | \le C \big(| M^{n-1} -  M^{n-2}|  + |\partial_x M^{n-1} - \partial_x M^{n-2}|\big) \,.
\]
Using \eqref{stgx1} we have,
\[
\begin{split}
|\partial_x M^{n-1} - & \partial_x M^{n-2}|  \le \bigg| \frac{\langle w \rangle_{f^{n-1},\varphi'}}{\langle 1 \rangle_{f^{n-1},\varphi}} - \frac{\langle w \rangle_{f^{n-2},\varphi'}}{\langle 1 \rangle_{f^{n-2},\varphi}}\bigg| \\ & \quad + \bigg|\frac{\langle w \rangle_{f^{n-1},\varphi}\langle 1 \rangle_{f^{n-1},\varphi'}}{\langle 1 \rangle_{f^{n-1},\varphi}^2} - \frac{\langle w \rangle_{f^{n-2},\varphi}\langle 1 \rangle_{f^{n-2},\varphi'}}{\langle 1 \rangle_{f^{n-2},\varphi}^2} \bigg| \\ & \le C \big(|\langle w \rangle_{f^{n-1},\varphi'} - \langle w \rangle_{f^{n-2},\varphi'}| + |\langle w \rangle_{f^{n-1},\varphi} - \langle w \rangle_{f^{n-2},\varphi}| \big)\\ & \quad + C |\langle w \rangle_{f^{n-2},\varphi}|\, \big(|\langle 1 \rangle_{f^{n-1},\varphi'} - \langle 1 \rangle_{f^{n-2},\varphi'}| + |\langle 1 \rangle_{f^{n-1},\varphi} - \langle 1 \rangle_{f^{n-2},\varphi}| \big) \\ & \le C\big(\|h^{n-1}-h^{n-2}\|_{L^1} + \|f^{n-1} - f^{n-2}\|_{L^1}\big)\,,
\end{split}
\]
where we used that $\varphi$ is strictly positive in the second inequality, and \eqref{suptirv} in the last inequality. 

From \eqref{contrl1}, \eqref{boh}, \eqref{duhampavf}, \eqref{duhampaxf}, and the above estimates, Eq.~\eqref{eqn:conclusion} follows. This concludes the proof. 
\end{proof}

\section{Particle system}
\label{sec:4}

In this section we work under Hypothesis \ref{standingassumptions}. Consider the  system of $N$ interacting particles, each of them having position ad valecity $(x_t^{i,N}, v_t^{i,N})$ described by the SDE \eqref{parsys1}-\eqref{parsys2}.  We recall that because $G$ is assumed to be bounded and Lipschitz continuous, existence and uniqueness (for every $N$ fixed) of the strong solution of the  system \eqref{parsys1}-\eqref{parsys2} follows from standard SDE theory.

\noindent
Let $\mathbb{P}\mathrm{r}_1$ be the space of probability measures on $\mathbb{T} \times \mathbb{R}$  with finite first moment and let
\[
\mathfrak{C}:= C([0,T]; \mathbb{P}\mathrm{r}_1)
\]
be the space of continuous functions from $[0,T]$ to the space $\mathbb{P}\mathrm{r}_1$. Consider the empirical measure, 
\[
S_t^N(\mathrm{d} x, \mathrm{d} v) := \frac 1N \sum_{i=1}^N \delta_{(x_t^{i,N}, v_t^{i,N})}(\mathrm{d} x, \mathrm{d} v)\,.
\]
The empirical measure is, for each fixed $t>0$ (and for each fixed $N \in \mathbb{N}$), a  random measure;  in particular, $S_t^N: \Omega \rightarrow \mathbb{P}\mathrm{r}_1$.\footnote{We will show that if the initial data have finite first moment, then the same property propagates at time $t>0$, for every fixed $N$; this can be seen directly from \eqref{parsys1}-\eqref{parsys2} or  Proposition \ref{uniquenessmeassol} and its proof.} Therefore, (for each fixed $N \in \mathbb{N}$ and $T>0$) the stochastic process $S^N = \{S_t^N\}_{t\in [0,T]} $ can be seen as a random variable with values in $\mathfrak{C}$. We will denote by $Q^N$ the law of the random variable $S^N$, so $\{Q^N\}_N$ is a sequence of probability measures on $\mathfrak{C}$.  As customary, if $\alpha$ is a measure on $\mathbb{T} \times \mathbb{R}$ and $\psi$ a function on the same space, we use in this section the notation
$$
\langle \alpha, \psi\rangle := \int_{\mathbb{T} \times \mathbb{R} } \psi(x,v) \, \alpha(\mathrm{d} x, \mathrm{d} v) \,.
$$
With this notation in place, the main result of this section is Theorem \ref{teo:convpartsys} below: roughly speaking,  as $N \rightarrow \infty$, the empirical distribution $S^N$ converges to the solution of equation \eqref{nl}.  

\begin{theorem}
\label{teo:convpartsys}
Suppose the initial data of the particle system \eqref{parsys1}-\eqref{parsys2} are such that
\begin{equation}
\label{conpsid}
\sup_{N\in\mathbb{N}} \mathbb{E} \max_{i=1,\ldots,N} \big| v^{i,N}_0 \big|^4 < +\infty
\quad \mbox{and} \quad \mathbb{E} \left[ \left\vert  \langle S_0^{N}, \psi \rangle - \langle f_0, \psi \rangle\right\vert \wedge 1 \right]  \stackrel{N \rightarrow\infty}{ \longrightarrow } 0, 
\end{equation}
where $f_0$ is the initial datum of \eqref{nl} and the above is supposed to hold for every $\psi \in C_0^{\infty}(\mathbb T \times \mathbb{R})$. 
 Then, with the notation introduced so far, the following holds: for each $t>0$, the sequence of random measures $\{S_t^N\}$ converges to a deterministic measure,  $S_t^*$, which has a density with respect to the Lebesgue measure. Such a density is a function in  the space $L^1(\sqrt{1+v^2}\, \mathrm{d} x \, \mathrm{d} v)$ and coincides with the unique solution of equation \eqref{nl} (given by Theorem \ref{mainthexun}). 

\end{theorem}

\begin{proof}
	
The proof is in four steps. 

\smallskip
\noindent
{\bf Step 1.} We start by finding the equation satisfied by $S_t^N$ (in weak sense). For any $\psi \in C_0^{\infty}(\mathbb{T} \times \mathbb{R})$, define
\[
\langle S_t^N, \psi \rangle := \iint_{\mathbb{T}\times \mathbb{R}} \frac 1N \sum_{i=1}^N \delta_{(x_t^{i,N}, v_t^{i,N})}(\mathrm{d} x, \mathrm{d} v) \psi(x,v) = \frac{1}{N} \sum_{i=1}^N \psi(x_t^{i,N}, v_t^{i,N})\,.
\]
Applying It\^o formula, one finds 
\begin{align}
\mathrm{d}\langle S_t^N, \psi \rangle & = 
\frac{1}{N} \sum_{i=1}^N(\partial_x\psi)(x_t^{i,N}, v_t^{i,N}) v_t^{i,N} \mathrm{d} t 
- \frac{1}{N} \sum_{i=1}^N(\partial_v\psi)(x_t^{i,N}, v_t^{i,N}) v_t^{i,N} \mathrm{d} t
\label{eqnextended1}\\
&\quad +\frac{1}{N} \sum_{i=1}^N (\partial_v\psi)(x_t^{i,N}, v_t^{i,N})
 G \left(\frac{\frac 1N \sum_{j=1}^N v_t^{j,N} \varphi(x_t^{i,N}- x_t^{j,N})}{\frac 1N \sum_{j=1}^N \varphi(x_t^{i,N}- x_t^{j,N})} \right) \mathrm{d} t \label{eqnextended2}\\
&\quad + \frac{\sigma}{N} \sum_{i=1}^N (\partial^2_v\psi)(x_t^{i,N}, v_t^{i,N}) \mathrm{d} t + \frac{\sigma}{N} \sum_{i=1}^N (\partial_v\psi)(x_t^{i,N}, v_t^{i,N})\, \mathrm{d} W_t^i \,. \label{eqnextended}
\end{align}
Letting 
$\varphi_x(\cdot) = \varphi (x - \cdot)$, we can write 
\[
\frac{1}{N} \sum_{j=1}^N \varphi(x - x_t^{j,N})
= \langle S_t^N, \varphi_{x} \rangle =: {\boldsymbol {{\varphi}}}^{S_t^N}(x)\,. 
\]
Similarly,  if $P$ is the function $P(x,v)=v$, we have,
\[
\frac{1}{N} \sum_{j=1}^N \varphi(x - x_t^{j,N}) v^{j,N}_t = \langle S_t^N, \varphi_{x} P\rangle =: {\boldsymbol {\tilde{\varphi}}}^{S_t^N}(x)\,. 
\]
More in general, if $\alpha$ is a measure, we define 
$$
{\boldsymbol {{\varphi}}}^{\alpha}(x) := \langle \alpha, \varphi_x \rangle
\quad  \mbox{and} \quad 
{\boldsymbol {\tilde{\varphi}}}^{\alpha}(x):= \langle \alpha, \varphi_x P\rangle \,.
$$
Therefore, 
\begin{equation}
\label{eqnforstn}
\begin{split}
\langle S_t^N, \psi \rangle & =\langle S_0^N, \psi \rangle + \int_0^t 
\mathrm{d} s \left [\langle S_s^N, (\partial_x \psi) P \rangle - \langle S_s^N, (\partial_v \psi) P \rangle  + \sigma \langle S_s^N, \partial^2_v \psi \rangle \right] \\ 
&  + \int _0^t \left\langle S_s^N,  (\partial_v \psi)\, G\left(\frac{{\boldsymbol {\tilde{\varphi}}}^{S_s^N} }{ {\boldsymbol {{\varphi}}}^{S_s^N}} \right)\right\rangle \mathrm{d} s +\mathcal{M}_t^{N, \psi}\,,
\end{split}
\end{equation}
having set 
\[
\mathcal{M}_t^{N, \psi}:=  \frac{\sigma}{N} \int_0^t \sum_{i=1}^N (\partial_v \psi)(x_s^{i,N}, v_s^{i,N}) \mathrm{d} W_s^i \,.
\]

\smallskip
\noindent
{\bf Step 2.} We want to show that the sequence of measures $\{Q^N\}_{N\in \mathbb{N}}$, or, equivalently, the sequence of random variables $\{S^N\}_{N\in \mathbb{N}}$,  is tight. To this end, we use \cite[Prop.~1.7]{KipnisLandim}. Namely, let $\{\psi_k\}_k$ be a dense subset of $C_b(\mathbb{T} \times \mathbb{R})$. Then the sequence $S^N$ is tight if and only if, for every $k \in \mathbb{N}$ the real-valued stochastic process
\begin{equation}
\label{realvalproc}
H_t^N:= \langle S_t^N, \psi_k \rangle
\end{equation}
is tight (it would be more correct to include the index $k$ in the notation for $H_t^N$, but we drop such an index to avoid cumbersome notations). In other words, \cite[Prop.~1.7]{KipnisLandim} reduces the problem of studying tightness of a family of probability measures on $C([0,T]; \mathbb{P}\mathrm{r}_1)$ to the simpler problem of studying tightness of a family of probability measures on $C([0,T]; \mathbb{R})$. We can now apply Kolmogorov's criterion to the process $H_t^N$; we therefore need to prove 
\begin{align}
& \textrm{i) } \sup_{N\in\mathbb{N}} \mathbb{E} \left\vert H_t^N - H_u^N \right\vert^p \leq C \left\vert t- u\right\vert^{\alpha} \quad \mbox{ for some $C>0$ and $p, \alpha>1$;} \label{Kolm}\\
& \textrm{ii) } \lim_{K \to \infty} \sup_{N\in\mathbb{N}} \mathbb{P} \left(\left\vert H_0^N- {a}\right\vert>K \right) =0 \quad  \mbox{for some $a \in \mathbb{R}$.}
\label{Kolm2}
\end{align}
This is the content of Lemma \ref{lemmaKolm}. Notice that condition \eqref{Kolm2} is automatically satisfied, for example for $a=0$, thanks to the boundedness of the functions $\psi_k$.  So the proof of Lemma \ref{lemmaKolm} is concerned with checking \eqref{Kolm}. 

\smallskip
\noindent
{\bf Step 3.} From step 2, every (sub-)sequence of $S^N$ admits a weakly convergent (sub-) sequence. We want to show that if $S^*=\{S_t^*\}_{t\geq 0}$ (with law $Q^*$) is the limit of any such subsequences, then $S^*$  is a weak measure-valued solution of Eq.~\eqref{nl} (in particular this implies that $S^*$ is deterministic). 
 We clarify that a path $\{\pi_t\} \subset \mathfrak{C}$ is a weak measure-valued solution of \eqref{nl} with initial datum $\pi_0 \in \mathbb{P}\mathrm{r}_1$ if the identity
\begin{equation}
\label{weakmeasureformulation}
\begin{split}
\langle \pi_t, \psi \rangle  & =  \langle \pi_0, \psi \rangle + \int_0^t\! \mathrm{d} s\, \left[ \langle \pi_s, (\partial_x \psi) P \rangle  - \langle \pi_s, (\partial_v \psi) P \rangle  + \sigma \langle \pi_s, \partial^2_v \psi \rangle  \right]  \\ & \quad + \int_0^t\! \mathrm{d} s\, \left\langle \pi_s,  (\partial_v \psi)\, G\left(\frac{{\boldsymbol {\tilde{\varphi}}}^{\pi_s} }{ {\boldsymbol {{\varphi}}}^{\pi_s}} \right)\right\rangle 
\end{split}
\end{equation}
is satisfied for every test function $\psi \in C_0^{\infty}(\mathbb{T} \times \mathbb{R})$. To this end,  for each $\psi \in C_0^{\infty}(\mathbb{T} \times \mathbb{R})$, consider the functional $J_{\psi}^{\pi_0} \colon \mathfrak{C} \longrightarrow \mathbb{R}_+$, defined as follows
\begin{equation}
\label{deffun}
\begin{split}
J_\psi^{\pi_0}(&\nu)   = \sup_{t \in [0,T]} \bigg| \langle \nu_t, \psi \rangle - \langle \pi_0, \psi \rangle  - \int_0^t\!\mathrm{d} s \big[ \langle \nu_s, (\partial_x \psi) P \rangle  \\ & \quad \, - \langle \nu_s, (\partial_v \psi) P \rangle + \sigma \langle \nu_s, \partial^2_v \psi \rangle  \big]  - \int_0^t\! \mathrm{d} s\, \left\langle \nu_s,  (\partial_v \psi)\, G\left(\frac{{\boldsymbol {\tilde{\varphi}}}^{\nu_s} }{ {\boldsymbol {{\varphi}}}^{\nu_s}} \right)\right\rangle  \bigg| \wedge 1\,.
\end{split}
\end{equation}
We stress that here $\pi_0$ is the initial datum for \eqref{weakmeasureformulation}.
The rationale underlying the choice of this functional can be understood by comparing with  \eqref{eqnforstn}: if the path $\nu=\{\nu_t\}_t$ is $S^N=\{S_t^N\}_t$, then 
$J_{\psi}^{\pi_0}(S^N) = \sup_{t \in [0,T]} \left\vert  \mathcal{M}_t^{N, \psi} \right\vert \wedge 1$; roughly speaking, the functional $J_{\psi}^{\pi_0}$ associates to a given path the martingale part of equation \eqref{eqnforstn}. We want to show that
\[
Q^*(\{\pi \in \mathfrak{C} : J_\psi^{\pi_0}(\pi)=0 \mbox{ for every } \psi \in C_0^{\infty}(\mathbb{T} \times \mathbb{R})\}) =1\,.
\]
Because $J_\psi^{\pi_0}$ is a positive functional, to prove the above it suffices to show that $\mathbb{E}^{Q^*}J_\psi^{\pi_0} = 0$ for every $\psi$. The functional $J_\psi^{\pi_0}$ is bounded and continuous, see Lemma \ref{lemmacontinuousfunctional}. Therefore, if $C>0$ is a generic positive constant, for every $\psi \in C_0^{\infty}(\mathbb{T} \times \mathbb{R})$, we have 
\[
\mathbb{E}^{Q^*} J_\psi^{\pi_0} \le  \liminf_{k\rightarrow \infty} \mathbb{E}^{Q^{N_k}} J_\psi^{\pi_0} \,.
\]
The right-hand side of the above is easy to estimate; indeed, by definition of $S^N$,
\begin{align*}
\mathbb{E}^{Q^{N}} J_{\psi}^{\pi_0} & = \mathbb{E} J_{\psi}^{\pi_0}(S^{N_k})  = 
\mathbb{E} \left[ \sup_{t \in [0,T]} \left\vert  \langle S_0^{N_k}, \psi \rangle - \langle \pi_0, \psi \rangle+ \mathcal{M}_t^{N, \psi} \right\vert \wedge 1 \right]\, \\
& \le  \mathbb{E} \left[ \left\vert  \langle S_0^{N_k}, \psi \rangle - \langle \pi_0, \psi \rangle\right\vert \wedge 1 \right] +  \left( \mathbb{E}  \sup_{t \in [0,T]} \left\vert  \mathcal{M}_t^{N, \psi} \right\vert^2 \right)^{1/2} \,.
\end{align*}
The first addend in  the above goes to zero by assumption, if we take $\pi_0$ to be a measure with density $f_0$.  Regarding the second:
\begin{align*}
  \left( \mathbb{E}  \sup_{t \in [0,T]} \left\vert  \mathcal{M}_t^{N, \psi} \right\vert^2 \right)^{1/2} & \le \frac{C}{N} \left( \mathbb{E} \left\vert \int_0^T \partial_v\psi(x_t^{i,N}, v_t ^{i,N})\, \mathrm{d} W^i\right\vert^2\right)^{1/2} \le \frac{C}{N} \stackrel{N \rightarrow \infty}{\longrightarrow }0 \,. 
\end{align*}

\smallskip
\noindent
{\bf Step 4.} As a consequence of Step 2 and Step 3, we know that Eq.~\eqref{nl} admits at least one measure-valued solution. It remains to prove that such a measure-valued solution is unique. Indeed (under the assumptions on the initial datum stated in Subsection \ref{sec:3.3}), this solution has to coincide with the solution in $L^1(1+\left\vert v\right\vert)$ that we found in Subsection \ref{sec:3.3}.\footnote{Notice that if $f_t(x,v)$ is a weak solution of \eqref{nl} in the sense \eqref{nonlintermwf}, then $\pi_t(\mathrm{d} x,\mathrm{d} v) = f_t(x,v)\, \mathrm{d} x\, \mathrm{d} v$ is a measure that satisfies \eqref{weakmeasureformulation}.} So we are left with showing uniqueness of measure-valued solutions of \eqref{nl}. This is the content of Proposition \ref{uniquenessmeassol}.  
\end{proof}

\begin{lemma}
\label{lemmacontinuousfunctional} 
For each function $\psi \in C_0^{\infty}(\mathbb{T} \times \mathbb{R})$, the functional $J_\psi^{\pi_0}$ (defined in \eqref{deffun}) is continuous. 
\end{lemma}

\begin{proof}
Let $\pi^N= \{\pi_t^N\}_t$ be sequence in $\mathfrak{C}$.  Suppose $\pi^N$ converges in $\mathfrak{C}$ to $\pi$. Then, for every $t>0$,  $\pi_t^N$ converges weakly to $\pi_t$. Hence, by definition $\langle \pi_t^N, \psi \rangle \rightarrow \langle \pi_t, \psi \rangle$. The same argument can be applied to all the other terms on the first line of \eqref{deffun}. 
To pass  to the limit in the nonlinear term (the last term in \eqref{deffun}),   we first act  with manipulations analogous to those in \eqref{mm} and then conclude with the same argument as above.  We omit the details.  
\end{proof}

\begin{lemma}
\label{lemmaKolm}
The process $H_t^N$ defined in \eqref{realvalproc} is tight on $C([0,T]; \mathbb{R})$. 
\end{lemma}

\begin{proof}
From the comments of Step 2, we only need to verify \eqref{Kolm}. We will show that this tightness bound is satisfied for  with $p=4$ (and, consequently,  $\alpha = 2$). 

An explicit expression for  the difference $\left\vert H_t^N-H_u^N\right\vert^4$ is found by using \eqref{eqnextended} (just integrate such an expression between $u$ and $t$ and then estimate from the above with the fourth power of each addend). The terms coming from the addends in \eqref{eqnextended2} and \eqref{eqnextended} can be easily estimated from above (independently of $N$) thanks to  
 the boundedness of $G$ and to the boundedness of all the derivatives of $\psi_k$.  As for the terms coming from the addends in \eqref{eqnextended1}, we observe that if the property \eqref{conpsid} holds for the initial datum of the particle system \eqref{parsys1}-\eqref{parsys2}, then the same property  holds at subsequent times as well. That is,  one can show that  $\mathbb{E} \sup_i \big| v^{i,N}_t \big|^4$ is bounded. With this observation in place, the conclusion of the proof follows. 
\end{proof}

\begin{proposition}
\label{uniquenessmeassol} 
There exists a unique measure-valued solution of Eq.~\eqref{weakmeasureformulation}. Moreover, if the initial datum $\pi_0$ has finite first moment, i.e. $\langle \pi_0, \left\vert v \right\vert \rangle < \infty$, then the solution $\pi_t$ has finite first moment as well. 
\end{proposition}

\begin{proof}
The existence claim is a result of Step 2 and Step 3, so we concentrate on proving uniqueness. This is done by using the same strategy adopted in \cite[Section 6]{FlandoliCapasso}, \cite[Sect.~3]{Flandoli} and \cite[Thm.~4.2]{Meleard}. Here we follow \cite[Sect.~6]{FlandoliCapasso} and \cite[Section 3]{Flandoli}  so we outline the strategy but we don't repeat all the details. 

Let $\mathcal{A}$ be the second order differential operator defined on smooth functions (of $x \in \mathbb{T}$ and $v \in \mathbb{R}$) as
\begin{equation}\label{cA}
 \mathcal{A}:= v \partial_x  - v \partial_v  + \sigma \partial_{vv}
\end{equation}
 and let $\mathcal{A}^*$ denote its formal adjoint (in the flat $L^2$ space). Then we can formally  rewrite \eqref{weakmeasureformulation} as 
\[
\pi_t= e^{t\mathcal{A}^*} \pi_0 + \int_0^t  \mathrm{d} s \,  e^{(t-s)\mathcal{A}^*} \left[ G\left(\frac{{\boldsymbol {\tilde{\varphi}}}^{\pi_s} }{ {\boldsymbol {{\varphi}}}^{\pi_s}}\right) (\partial_v \pi_s) \right] \,. 
\]
That is, 
\begin{equation}\label{solmild}
\langle \pi_t, \psi\rangle = 
\langle \pi_0 ,  e^{t\mathcal{A}} \psi\rangle  - \int_0^t \! \mathrm{d} s \, \left\langle  G\left(\frac{{\boldsymbol {\tilde{\varphi}}}^{\pi_s} }{ {\boldsymbol {{\varphi}}}^{\pi_s}}\right) \pi_s, \partial_v  
\left( e^{(t-s) \mathcal{A}} \psi  \right) \right\rangle \,.
\end{equation}
Let us now introduce the following metric on $\mathbb{P}\mathrm{r}_1$:
\[
d(\alpha, \beta):= \sup_{\|\psi\|_{\infty}\leq 1} \left\vert \langle \alpha-\beta,  \psi \rangle \right\vert\,,
\]
where the supremum is taken over all measurable bounded functions $\psi: \mathbb{T} \times \mathbb{R} \rightarrow \mathbb{R} $ with $\|\psi\|_{\infty}\leq 1$. Notice that the metric $d$ can be equivalently defined as 
\[
d(\alpha, \beta):= \sup_{\psi \in C_0^{\infty}, \|\psi\|_{\infty}\leq 1} \left\vert \langle \alpha-\beta,  \psi \rangle \right\vert\,. 
\]
Showing the equivalence of the two definitions  is a simple application of the dominated convergence theorem (see \cite[Remark 3.2]{Flandoli}). 
Let $\pi= \{\pi_{t}\}_{t\geq 0}$ and $\nu= \{\nu_{t}\}_{t\geq 0}$ be two solutions of \eqref{weakmeasureformulation}, with the same initial datum. The aim of the rest of the proof is to close an integral inequality on the quantity $\sup_{t\in [0,T]} d(\pi_t, \nu_t)$, thereby showing that $\pi_t=\nu_t$ for every $t \in [0,T]$. However (similarly to what happened in  the calculations of Section \ref{sec:3.3}) closing the inequality on $d(\pi_t, \nu_t)$ is not possible. For this reason we will instead work with the quantity
\[
\tilde d (\pi_t, \nu_t) := d(\pi_t, \nu_t) + \sup_{\|\psi\|_{\infty}\leq 1} \langle  \pi_t - \nu_t  ,\left\vert v \right\vert \psi \rangle \,.
\]
We will repeatedly use the inequality
\begin{equation}
\label{pullsup}
\left\vert \langle \pi_t-\nu_t, (1+ \left\vert v \right\vert) \psi \rangle \right\vert \le \|\psi\|_{\infty} \tilde{d}(\pi_t, \nu_t) \,.
\end{equation}
We will show in some detail how to study the term $\langle \pi_t-\nu_t,  \psi \rangle $, the term $\langle \pi_t-\nu_t,  \left\vert v \right\vert \psi \rangle $ can then be treated analogously (and we will only point out the slight difference, without redoing the calculation). Throughout the proof $C>0$ will be a generic constant. From \eqref{solmild} we have,
\[
|\langle \pi_t-\nu_t, \psi \rangle| \le I_1+ I_2\,,
\]
where
\begin{align*}
I_1 & = \int_0^t\! \mathrm{d} s \,  \left\vert \left\langle  \left( G\left(\frac{{\boldsymbol {\tilde{\varphi}}}^{\nu_s} }{ {\boldsymbol {{\varphi}}}^{\nu_s}}\right) - G \left(\frac{{\boldsymbol {\tilde{\varphi}}}^{\pi_s} }{ {\boldsymbol {{\varphi}}}^{\pi_s}}\right)  \right)  \nu_s , \partial_v \left( e^{(t-s) \mathcal{A}} \psi  \right)  \right\rangle \right\vert\,,  \\ I_2  & =  \int_0^t\! \mathrm{d} s \,  \left\vert \left\langle G \left(\frac{{\boldsymbol {\tilde{\varphi}}}^{\pi_s} }{ {\boldsymbol {{\varphi}}}^{\pi_s}}\right) (\nu_s-\pi_s), \partial_v \left( e^{(t-s) \mathcal{A}} \psi  \right)  \right\rangle \right\vert\,. 
\end{align*}
In view of \eqref{pullsup}, the term $I_2$ is readily bounded,
\[
I_2 \le \|G\|_\infty \int_0^t\! \mathrm{d} s\,  \big\|\partial_v \big( e^{(t-s) \mathcal{A}} \psi\big)\big\|_\infty \tilde{d}(\nu_s, \pi_s) \leq C  \int_0^t\! \mathrm{d} s\, \frac{\tilde{d}(\nu_s, \pi_s)}{\sqrt{t-s}}\,,
\]
having used, in the last inequality, the following heat kernel type of bound
\begin{equation}
\label{heatkernel1}
\partial_v \left(e^{t\mathcal{A}} \psi \right) \leq \frac{C}{\sqrt{t-s}} \|\psi\|_{\infty}\,.
\end{equation}
The above estimate is classical and holds for every continuous and bounded function $\psi$, see for example \cite{OttobreZegarlinski} or, for a more explicit proof adapted to this case, see \cite[Lemma 12]{FlandoliCapasso}. As for the term $I_1$, after using the lipschitzianity of $G$ and manipulations similar to those in \eqref{mm}, one finds
\[
I_1 \leq C (I_{1,1} + I_{1,2})\,,
\]
where
\begin{align*}
I_{1,1}  & := \int_0^t\! \frac{\mathrm{d} s}{\sqrt{t-s}} \iint_{\mathbb{T}\times \mathbb{R}}\! \left\vert 
{\boldsymbol {\tilde{\varphi}}}^{\pi_s} - {\boldsymbol {\tilde{\varphi}}}^{\nu_s} \right\vert
\nu_s(\mathrm{d} x,\mathrm{d} v)  \\ & = \int_0^t\! \frac{\mathrm{d} s}{\sqrt{t-s}} \int_{\mathbb{T}\times \mathbb{R}}\! \nu_s (\mathrm{d} x,\mathrm{d} v) \left\vert \int_{\mathbb{T}\times \mathbb{R}}\! \varphi(x-y) \, w \, (\pi_s-\nu_s)(\mathrm{d} y,\mathrm{d} w) \right\vert
\end{align*}
and 
\begin{align*} 
I_{1,2} := \int_0^t\! \frac{\mathrm{d} s}{\sqrt{t-s}}  \int_{\mathbb{T}\times \mathbb{R}}\! \left\vert {\boldsymbol {\tilde{\varphi}}}^{\nu_s}(x) \right\vert \nu_s(\mathrm{d} x, \mathrm{d} v) \left\vert  \int_{\mathbb{T}\times \mathbb{R}}\! \varphi(x-y) (\nu_t -\mu_t)(\mathrm{d} y,\mathrm{d} w) \right\vert\,.
\end{align*}
If we multiply and divide the integral of the inner integrand in the definition of $A_1$, we then obtain 
\[
I_{1,1} \le \int_0^t\! \frac{\mathrm{d} s}{\sqrt{t-s}} \tilde{d}(\pi_s, \nu_s) \,.
\]
In order to estimate the term $I_{1,2}$, all we need to do is to find a bound on ${\boldsymbol {\tilde{\varphi}}}^{\nu_s}(x)$. This can be done by using the exact same approach that we have used so far. In particular, it suffices to show that $\langle \pi_t , \left\vert v \right\vert \rangle \leq \kappa(t)$ where $\kappa(t)$ is a generic function bounded on compacts.  Again we write,
\[
\langle\pi_t, (1+\left\vert v \right\vert) \psi \rangle = \langle\pi_0, \left\vert v \right\vert \psi \rangle - \int_0^t\! \mathrm{d} s\, \left \langle G\left( \frac{{\boldsymbol {\tilde{\varphi}}}^{\pi_s}}{{\boldsymbol {{\varphi}}}^{\pi_s}} \right) \pi_s, \partial_v \left( e^{(t-s)\mathcal{A}} (1+\left\vert v \right\vert) \psi \right) \right\rangle\,.
\]
We now use the estimate 
\begin{equation}
\label{heatkernel2}
\partial_v \left(e^{t\mathcal{A}} \left\vert v \right\vert \psi \right) \leq \frac{C}{\sqrt{t-s}} \|\psi\|_{\infty} (1+ \left\vert v \right\vert) \,
\end{equation}
(again, a constructive proof of the above can be found in \cite[Lemma 12]{FlandoliCapasso}) and \eqref{heatkernel1} to  conclude
\begin{equation}\label{mommeas}
\langle\pi_t, (1+\left\vert v \right\vert) \psi \rangle  \le  \langle\pi_0, (1+\left\vert v \right\vert)\, \psi \rangle + \int_0^t\! \mathrm{d} s\, \frac{1}{\sqrt{t-s}}\|G\|_\infty \|\psi\|_\infty  \iint (1+\left\vert v \right\vert) \pi_s (\mathrm{d} x, \mathrm{d} v)\,. 
\end{equation}
By taking $\psi \equiv 1$, one gets $\iint_{\mathbb{T}\times\mathbb{R}} \left\vert v \right\vert \pi_s (\mathrm{d} x, \mathrm{d} v) \le C e^{Ct}$.  Hence 
\[
I_1\le C \int_0^t\!\mathrm{d} s\, \frac{e^{Ct}}{\sqrt{t-s}} \tilde{d}(\pi_s, \nu_s) \,.
\]
The proof is concluded after use of the (generalised) Gronwall Lemma. 
\end{proof}

\section{Invariant measures: some comments}
\label{sec:5}

The non-homogeneous stationary problem reads,
\begin{equation}
\label{ns1}
v \partial_x f + G(M(x)) \partial_v f = \partial_v(\sigma\partial_v f + vf)\,, \quad M(x) := \frac{\int_{\mathbb{T}_L}\!\mathrm{d} y \int_{\mathbb{R}}\!\mathrm{d} w\, f(y,w)\,\varphi(x-y)\, w}{\int_{\mathbb{T}_L}\!\mathrm{d} y \int_{\mathbb{R}}\!\mathrm{d} w\, f(y,w)\,\varphi(x-y)}\,.
\end{equation}
Assuming that the boundary terms do not contribute, integration with respect to the velocity gives $\partial_x \int_{\mathbb{R}}\!\mathrm{d} v f(x,v)\, v =0$. Hence the local average velocity $\int_{\mathbb{R}}\!\mathrm{d} v f(x,v)\, v$ does not depend on $x$; we denote it by $\alpha$. Moreover, as $\int_{\mathbb{T}_L}\!\mathrm{d} y\,\varphi(y)=1$, 
\[
\int_{\mathbb{T}}\!\mathrm{d} y \int_{\mathbb{R}}\!\mathrm{d} w\, f(y,w)\,\varphi(x-y)\, w = \alpha \int_{\mathbb{T}}\!\mathrm{d} y\, \varphi(x-y) = \alpha\,,
\]
so that, setting $\varrho(x) = \int_{\mathbb{R}}\!\mathrm{d} w\, f(x,w)$, 
\[
M(x) = \frac{\alpha}{(\varphi*\varrho)(x)}\,,
\]
where $*$ denotes convolution on $\mathbb{T}$. 

Suppose now that the interaction $\varphi$ is close to the constant function, i.e., $\varphi = 1 + \lambda \psi$ with $0<\lambda \ll 1$, $\|\psi\|_{L^\infty}=1$, and $\int_{\mathbb{T}}\mathrm{d} x\, \psi(x)  = 0$. Eq.~\eqref{ns1} then reads,
\begin{equation}
\label{ns2}
v \partial_x f + G\bigg(\frac{\alpha}{1+\lambda\psi*\varrho}\bigg) \partial_v f = \partial_v(vf+\sigma\partial_v f )\,.
\end{equation}
As $G(0)=0$, if $\alpha=0$ the equation reduces to the linear Langevin equation with vanishing force, hence the Gaussian distribution $\mathcal{N}(0,\sigma)$ gives the (unique) solution. Let us consider the case $\alpha\ne 0$. By symmetry, without loss of generality we can restrict to the case $\alpha>0$. 

\begin{theorem}
\label{thm:nsl}
Let $f(x,v,\lambda)$ be a solution to \eqref{ns2} with $\alpha(\lambda)=\int_{\mathbb{R}}\!\mathrm{d} v f(x,v,\lambda)\, v >0$. Suppose the following holds: 
\begin{itemize}
\item[(i)] The density $f$ is smooth and admits a convergent expansion in $\lambda$,
\begin{equation}
\label{fex}
f(x,v,\lambda) = \sum_{n=0}^\infty f_n(x,v)\,\lambda^n \,.
\end{equation}
\item[(ii)] The coefficients $f_n(x,v)$ are differentiable in $x,v$ and the $x,v$-derivatives of $f$ can be computed by term by term differentiation, obtaining convergent series. Moreover, $f_n\in L^2(\mathbb{T}_x;\mathcal{H}_v^{1,2})$.
\end{itemize}
Then $f$ coincides with the Gaussian density $\mathcal{N}(1,\sigma)$.
\end{theorem}

\begin{proof}
We have to prove that  
\[
f_0 = \mathcal{N}(1,\sigma), \quad f_n=0 \;\;\forall\, n\ge 1\,.
\]
This is achieved by induction on $n$. We start with the inductive basis, which is the case $n=0$. By \eqref{fex} we have,
\[
\alpha = \alpha(\lambda) = \sum_{n=0}^\infty \alpha_n \, \lambda^n\,, \qquad \alpha_n = \int_{\mathbb{R}}\!\mathrm{d} v f_n(x,v)\, v\,.
\]
Plugging \eqref{fex} and the above in \eqref{ns2}, after equating the $0$-th order in $\lambda$ on both sides we get,
\[
v \partial_x f_0 + G(\alpha_0) \partial_v f_0 = \partial_v( vf_0+\sigma\partial_v f_0)\,,
\]
whose solution is $f_0 = \mathcal{N}(G(\alpha_0),\sigma)$. Finally, since $\alpha_0 =  \int_{\mathbb{R}}\!\mathrm{d} v f_0(x,v)\, v = G(\alpha_0)$, we conclude that $\alpha_0=1$, i.e., $f_0 = \mathcal{N}(1,\sigma)$. 

We now turn to the inductive step: given $n>1$, prove that if $f_0= \mathcal{N}(1,\sigma)$ and $f_k=0$ for any $1\le k \le n-1$, then $f_n=0$. Under the inductive hypothesis and recalling that $\psi$ has zero average,
\[
\alpha = 1 + \sum_{k\ge n} \alpha_k \lambda^k\,, \qquad 1+\lambda\psi*\rho = 1 + \sum_{k\ge n} \lambda^{k+1}\psi*\rho_k\,, \qquad \rho_k(x) := \int_{\mathbb{R}}\!\mathrm{d} v\, f_k(x,v)\,.
\]
Moreover, because $f_0$  has mass one,  
\[
\sum_{k\ge n} \lambda^k \int_{\mathbb{T}}\!\mathrm{d} x\int_{\mathbb{R}}\!\mathrm{d} v\,  f_k(x,v) = 0 \quad\Longrightarrow \quad  \int_{\mathbb{T}}\!\mathrm{d} x\int_{\mathbb{R}}\!\mathrm{d} v f_k(x,v) = 0 \quad\forall\, k\ge n\,.
\]
Plugging the expansions \eqref{fex} in \eqref{ns2} and equating the $n$-th order terms in $\lambda$ we obtain,
\[
v \partial_x f_n + \partial_v f_n + \alpha_nG'(1) \partial_v f_0 = \partial_v (vf_n + \sigma\partial_v f_n)\,.
\]
We multiply both sides of the above equation by $v$, and then we integrate on $\mathbb{T}_x\times\mathbb{R}_v$. Using that $\int_{\mathbb{T}}\!\mathrm{d} x\int_{\mathbb{R}}\!\mathrm{d} v\, f_n(x,v) = 0$, after an integration by parts we obtain the identity $-\alpha_n G'(1) = -\alpha_n$, which implies $\alpha_n=0$ because $G'(1)\ne 1$. Therefore,
\[
v \partial_x f_n + \partial_v f_n  = \partial_v(vf_n + \sigma\partial_v f_n)\,.
\]
Now observe that $f_n=0$ is the unique solution, in $L^2(\mathbb{T}_x;\mathcal{H}_v^{1,2})$ and with zero average, of the above equation.\footnote{See e.g.  \cite{Villani}.} This concludes the proof. 
\end{proof}

\appendix

\section{Analysis of the non-homogeneous linear equation}
\label{AppA}

We here study Eq.~ \eqref{genlineqn}. In what follows, we denote by $A_m(v)$ the derivative of $\sqrt{1+v^{2m}}$, i.e.,
\begin{equation}\label{defAM}
A_m(v)= \frac{m v^{2m-1}}{\sqrt{1+v^{2m}}}\,.
\end{equation}
For any function of $v$, say $h(v)$, we write $h \sim v^p$ to indicate that  asymptotically $h$ grows like $v^p$. I.e., we write $h \sim v^p$ if $\lim_{\left\vert v \right\vert \rightarrow\infty}\frac{h(v)}{v^p} = \mathrm{const}$.   We  also (improperly)  use the notation
\[
\langle f, g \rangle_{(X^m)^*, X^m}=\int_0^T\! \mathrm{d} t \!\int_{\mathbb{T}}\! \mathrm{d} x \!\int_{\mathbb{R}}\! \mathrm{d} v f g \sqrt{1+v^{2m}} \,. 
\]
\begin{remark}
\label{noteflat}
We point out the simple fact - which we will use more or less explicitly in the following, sometimes without mention -  that if $g$ solves the weighted variational problem \eqref{wvp} for some $m>0$,  then it also solves the flat problem $E^0(g, \phi)=L^0 (\phi)$ (to see this just consider $\tilde{\phi}=\phi \sqrt{1+v^{2m}} $).
\end{remark}

\begin{proof} [Proof of Proposition \ref{propwellposlin}] 
We follow the approach in  \cite[Prop.~A.1]{Degond}, so we don't repeat the whole proof but just highlight the differences. To make comparison easier, we try to use a notation similar to the one in \cite{Degond}. 
\noindent
If $g$ solves \eqref{sf},  after the change of unknown  $u(t,x,v)= e^{-\lambda t } g(t,x,v)$, the function $u$ solves the equation
\begin{equation}
\label{UU}
\partial_t u  +v \partial_x u + a(t,x,v) \partial_v u+ (\lambda+c-1) u - \sigma \partial_{vv}u = U(t,x,v) e^{-\lambda t}\,.
\end{equation}
In the following we will simply denote by $U$ the function $Ue^{-\lambda t}$. The weak formulation of the problem  \eqref{UU} is given by 
\[
E^m_{\lambda}(u, \phi) = L^m(\phi), \quad \phi \in \Phi\,, 
\]
where
$$
E^m_{\lambda}(u, \phi)= E^m(u, \phi )+ \int_0^T\!\mathrm{d} s\int_{\mathbb{T}}\!\mathrm{d} x \int_{\mathbb{R}}\!\mathrm{d} v\, \lambda \, u \, \phi   \sqrt{1+v^{2m}}\,, 
$$
and $E^m, L^m$ and the space $\Phi$ have been defined at the beginning of Subsection \ref {sec:3.2}. 
After observing that for any $m\geq 0$ and $\phi \in \Phi$ one has 
\begin{equation}\label{trick}
	\begin{split}
2 \int_{\mathbb{T}}\!\mathrm{d} x \int_{\mathbb{R}}\!\mathrm{d} v\,  \sqrt{1+v^{2m}} \phi \, v \, \partial_v \phi & = - \int_{\mathbb{T}}\!\mathrm{d} x \int_{\mathbb{R}}\!\mathrm{d} v\,  \sqrt{1+v^{2m}} \phi^2 \\ & \quad -  \int_{\mathbb{T}}\!\mathrm{d} x \int_{\mathbb{R}}\!\mathrm{d} v\,  \frac{m v^{2m}}{\sqrt{1+v^{2m}}} \phi^2 \,,
\end{split}
\end{equation}
it is easy to show that the bilinear form $E^m_{\lambda}$  is also coercive (i.e., there exists a constant $C>0$ such that $E_{\lambda}(\phi, \phi) \geq C \|\phi\|^2_{X^m}$, when we choose $\lambda>1$  large enough, depending on $\|b\|_{\infty}$ and $c$) and continuous on $X^m$, for every fixed $\phi \in \Phi$.\footnote{This can be done with calculations very similar to those that we will show in  detail below to prove uniqueness, so we omit them.} Therefore Lions' Theorem \cite[page 534]{Degond} can be  applied and gives existence of a solution in ${X}^{m}$. 
As for uniqueness, we show this part in some detail; we follow \cite[pages 535-536]{Degond},  but focus on uniqueness in the space $X^m$: we need to prove that if $u_0=0$ and $U=0$ then $u \equiv 0$. First of all, one needs to show that the space $\tilde{\Phi}$ of smooth functions with compact support in $[0,T] \times \mathbb{T}_x \times \mathbb{R}_v$ is dense in $\mathcal{Y}^m$ (with the topology induced by the $X^m$-norm).\footnote{The density of $C_0^{\infty}$ in $\mathcal{H}^{1, m}_v$ (in the norm of $\mathcal{H}^{1, m}_v$) is an obvious consequence of  the density of $C_0^{\infty}$ in $H^1$; indeed, if $f\in \mathcal{H}^{1, m}_v$ then $f (1+v^{2m})^{1/4} \in H^1$.} This can be done following the proof of \cite[Lemma A.1]{Degond}, so we don't repeat it. With this technical detail in place, it is straightforward to see that for any two functions $u, \tilde{u} \in \mathcal{Y}^m$ (just functions of $\mathcal{Y}^m$, not necessarily solutions of our problem) one has,
\begin{equation}
\label{bt}
\begin{split}
&\langle \partial_t u + v \partial_xu, \tilde{u} \rangle_{(X^m)^*, X^m} + \langle \partial_t \tilde{u} + v \partial_x\tilde{u}, u \rangle_{(X^m)^*, X^m}\\
& = \int_{\mathbb{T}}\!\mathrm{d} x \int_{\mathbb{R}}\!\mathrm{d} v\,  u(T,x,v) \tilde{u}(T,x,v)\sqrt{1+v^{2m}} \\ & \quad - \int_{\mathbb{T}}\!\mathrm{d} x \int_{\mathbb{R}}\!\mathrm{d} v\,  u(0,x,v) \tilde{u}(0,x,v)\sqrt{1+v^{2m}}\,. 
\end{split}
\end{equation}
So, if $u$ is a solution of our problem with $u_0=0$ and $U=0$ then 
\begin{align}
0 &= \langle \partial_t u + v \partial_xu, {u} \rangle_{(X^m)^*, X^m} + \int_0^T \int_{\mathbb{T}}\!\mathrm{d} x\! \int_{\mathbb{R}}\!\mathrm{d} v\,  (\lambda+c-1) u^2 \sqrt{1+v^{2m}} \nonumber\\
&\quad+ \int_0^T \int_{\mathbb{T}}\!\mathrm{d} x\! \int_{\mathbb{R}}\!\mathrm{d} v\,  (\partial_v u) \, b \, u \sqrt{1+v^{2m}} - \int_0^T \int_{\mathbb{T}}\!\mathrm{d} x\! \int_{\mathbb{R}}\!\mathrm{d} v  (\partial_v u)  v u \sqrt{1+v^{2m}} \label{a4}\\
&\quad+ \sigma \int_0^T \int_{\mathbb{T}}\!\mathrm{d} x\! \int_{\mathbb{R}}\!\mathrm{d} v\,  \left\vert \partial_v u\right\vert^2 \sqrt{1+v^{2m}} 
+ \sigma \int_0^T \int_{\mathbb{T}}\!\mathrm{d} x\! \int_{\mathbb{R}}\!\mathrm{d} v\,  \partial_v u \, u \frac{mv^{2m-1}}{\sqrt{1+v^{2m}}}\,. 
\label{a5}
\end{align}
Using \eqref{trick} one can see that the second term in \eqref{a4} is positive:
\[
\begin{split}
 - \int_0^T \int_{\mathbb{T}}\!\mathrm{d} x\! \int_{\mathbb{R}}\!\mathrm{d} v\, &  (\partial_v u) \, v u \sqrt{1+v^{2m}}  \\ &  = \frac 12 \int_0^T \int_{\mathbb{T}}\!\mathrm{d} x\! \int_{\mathbb{R}}\!\mathrm{d} v\, \bigg[ \sqrt{1+v^{2m}} u^2 + \frac{m v^{2m}}{\sqrt{1+v^{2m}}} u^2 \bigg]  \ge 0\,.
\end{split}
\]
Notice that the integration by part used to obtain the above  equality is allowed and gives zero boundary terms because, being the initial datum $u_0=0$ in  $L^{2,m}$ for every $m>0$, the function $u$ is in ${X}^m$ for every $m>0$. Therefore $u$ and $vu\sqrt{1+v^{2m}}$ are in $H^1$, so that the boundary terms in the integration by parts disappear.  As for the first term in \eqref{a4}, using the so-called  Young's inequality with $\epsilon$ one has,\footnote{By acting in this way, one avoids to do integration by parts in this term, so the boundedness of $\partial_vb$ is not needed.}
\begin{equation}
\label{avoid}
\begin{split}
\int_0^T \int_{\mathbb{T}}\!\mathrm{d} x\! \int_{\mathbb{R}}\!\mathrm{d} v\,  (\partial_v u) \, b \, u \sqrt{1+v^{2m}} & \ge  - \frac{\epsilon\| b\|_{\infty} }{2}\int_0^T \int_{\mathbb{T}}\!\mathrm{d} x\! \int_{\mathbb{R}}\!\mathrm{d} v  (\partial_v u)^2 \sqrt{1+v^{2m}} \\ & \quad - \frac{\| b\|_{\infty}}{2\epsilon} \int_0^T \int_{\mathbb{T}}\!\mathrm{d} x\! \int_{\mathbb{R}}\!\mathrm{d} v \, u^2 \sqrt{1+v^{2m}}\,.
\end{split}
\end{equation}
The last term in \eqref{a5} can be estimated analogously. Choosing $\epsilon$ small enough (and possibly $\lambda$ large enough) and putting  everything together one obtains $0\geq \, d \| u\|_{X^m},$ where the constant $d>0$ depends on $\lambda, \epsilon, \|b\|_{\infty},c$ and $\sigma$. This implies $\| u\|_{X^m}=0$. 

Finally, if $u_0 \in L^{2,m}$ (which implies that $u \in X^m$) then, by definition, $ b \partial_v u$,  $\partial_{vv} u$, $u \in (X^m)^*$; as for $v\partial_vu$, this is in $(X^1)^* \subset (X^m)^*$.  Overall, one has 
\[
\partial_t u  +v \partial_x u =- b(t,x,v)  \partial_v u+ \partial_v(vu) -c\, u+ \sigma \partial_{vv}u + U(t,x,v) \in (X^{m})^*\,.
\]
The proposition is thus proved.
\end{proof}

\begin{proof}[Proof of Lemma \ref{lemmamaxprinc}] We just need to prove the result when $c=0$ (see point (i) of Remark \ref{rembeg}). If $c=0$
the proof is similar to what has been done in  \cite[pages 538-539]{Degond} (and also in our case it can indeed be done in a flat space, see Remark  \ref{noteflat}), so we don't repeat the details. We just stress two facts: (a) arguing analogously to what we have done in the  proof of uniqueness,  in order to do integration by parts in the term 
\[
-\int_{\mathbb{T}}\!\mathrm{d} x \int_{\mathbb{R}}\!\mathrm{d} v\,  v \partial_v u \, u^- \,,
\]
so to use a trick similar to \eqref{trick}, one needs to know that $vu \in H^1$, hence the assumption that the initial datum should be in $L^{2,m}$, $m \geq 2$; (b) the term containing $b$ should be handled in a similar way to what we have done in \eqref{avoid}, so to avoid having to assume  boundedness of $\partial_vb$. 
\end{proof}

\begin{proof}[Proof of Lemma \ref{lemmaduhamel}]
Again, we just need to prove the result when $c=0$ (see point (i) of Remark  \ref{rembeg}). If $c=0$ the statement can be proved by acting similarly to what has been done in \cite[Prop.~A.4]{Degond}, and we will therefore be, once again,  brief. We stress that, like in \cite{Degond}, for this proof we work in a flat space, with flat scalar products (see Remark \ref{noteflat}).  In \cite{Degond} the sought result is proved under the assumption that the coefficient $a(t,x,v)$ of Eq.~\eqref{UU} has zero $v$-divergence. Moreover,  the $x$-domain is the whole of $\mathbb{R}$.  In particular, the proof of the Duhamel formula relies on \cite[Eqs.~(60)-(62)]{Degond}. \cite[Eqs.~(60) and (62)]{Degond} remain unaltered (even if we change the $x$-domain). As for \cite[Eq.~(61)]{Degond}, upon inspecting the strategy of proof used in \cite{Degond}, the result follows if one proves that the left-hand side of \cite[Eq.~(61)]{Degond} is non-negative (it doesn't need to vanish). In our case, Eq.~(61) reads,\footnote{We recall that in \cite{Degond} the notation $\langle f,g \rangle_{X',X}$ simply stands for the pairing $\langle f,g \rangle_{X',X}= \int_0^T\!\mathrm{d} t \int_{\mathbb{T}}\!\mathrm{d} x \int_{\mathbb{R}}\!\mathrm{d} v\, f(t,x,v) g(t,x,v)$.}
\begin{equation}
\label{decomp}
\begin{split}
\int_0^T\!\mathrm{d} s\int_{\mathbb{T}}\!\mathrm{d} x \int_{\mathbb{R}}\!\mathrm{d} v\, a (\partial_v u) \psi_{\epsilon}(u) & =  \int_0^T\!\mathrm{d} s\int_{\mathbb{T}}\!\mathrm{d} x \int_{\mathbb{R}}\!\mathrm{d} v\,  b (\partial_v u) \psi_{\epsilon}(u)  \\ & \quad - \ \int_0^T\!\mathrm{d} s\int_{\mathbb{T}}\!\mathrm{d} x \int_{\mathbb{R}}\!\mathrm{d} v\, v (\partial_v u) \psi_{\epsilon}(u)\,, 
\end{split}
\end{equation}
where $\psi_{\epsilon}$ is defined like in \cite[page 540]{Degond}. That is, $\psi_{\epsilon}(u)$ is a smooth function of $u$ with $\psi_{\epsilon}(u)=0$ if $u \leq 0$,  $\psi_{\epsilon}(u)=1$ if $u \geq \epsilon$ and $\psi_{\epsilon}$ is increasing in $[0,\epsilon]$.  Letting $\varphi_{\epsilon}$ be a primitive of $\psi_{\epsilon}$, i.e.,
\begin{equation}
\label{prim}
\varphi'_{\epsilon}:= \psi_{\epsilon}\,,
\end{equation}
one has that $\varphi_{\epsilon}(u)$ is non-negative, i.e., $\varphi_{\epsilon}(u) \geq 0$\footnote{This is true by the definition of $\psi_{\epsilon}$.} and it  converges weakly to $u^+$. Using \eqref{prim},  the first term on the right-hand side of \eqref{decomp} becomes,
\[
\begin{split}
\int_0^T\!\mathrm{d} s\int_{\mathbb{T}}\!\mathrm{d} x \int_{\mathbb{R}}\!\mathrm{d} v\, b\, (\partial_v u) \psi_{\epsilon}(u)  & =  \int_0^T\!\mathrm{d} s\int_{\mathbb{T}}\!\mathrm{d} x \int_{\mathbb{R}}\!\mathrm{d} v\, b\, \partial_v(\varphi_{\epsilon}(u)) \\ & = -  \int_0^T\!\mathrm{d} s\int_{\mathbb{T}}\!\mathrm{d} x \int_{\mathbb{R}}\!\mathrm{d} v\, (\partial_v b) \varphi_{\epsilon}(u)\,.
\end{split}
\]
The boundary terms in the above integration by parts vanish thanks to the properties of the solution $u$, which is in $H^1_v$,  so that overall the right-hand side is non-negative (as $\partial_vb\leq 0$). The second term in \eqref{decomp} can be treated similarly:  integrating by parts\footnote{In this integration by parts the boundary terms disappear as $uv\in H^1$ when the initial datum is at least in $L^{2,2}$ and $\psi_{\epsilon}(u)$ is bounded.} we have,
\[
- \int_0^T\!\mathrm{d} s\int_{\mathbb{T}}\!\mathrm{d} x \int_{\mathbb{R}}\!\mathrm{d} v\, v (\partial_v u) \psi_{\epsilon}(u) = \int_0^T\!\mathrm{d} s\int_{\mathbb{T}}\!\mathrm{d} x \int_{\mathbb{R}}\!\mathrm{d} v\, \varphi_{\epsilon}(u)\,.
\]
Again,  the right-hand side   is positive. This is enough to conclude the proof. 
\end{proof}

\begin{proof}[Proof of Lemma \ref{lemmaequilim}] We  prove the estimates \eqref{fest}, \eqref{sest}, and \eqref{dest}. Throughout this proof $C>0$ will be a generic constant, so the value of $C$ may change from line to line. 

\smallskip
\noindent \textit{Proof of (\ref{fest}).}  After some calculations, one finds that the function $Y^{n,m}_t$ solves the PDE,
\begin{equation}
\label{Yeq}
\partial_t Y^{n,m}_t + v \partial_xY^{n,m}_t + B^{n,m}_t \partial_vY^{n,m}_t -\partial_v(vY^{n,m}_t)-\sigma \partial_{vv} Y^{n,m}_t = R^{n,m}_t\,, 
\end{equation}
where
\begin{equation}
\label{Yeqr}
B^{n,m}_t = G(M^{n-1}(t,x)) + 2 \sigma \frac{A_m(v)}{\sqrt{1+v^{2m}}}\,,\;\; R^{n,m}_t = \left[A_m(v)(B^{n,m}_t - v)  - \sigma A_m'(v)\right] f^{n}_t \, ,
\end{equation}
and $A_m$ is defined in \eqref{defAM}. 
We first check that \eqref{Yeq} is well posed, by applying Proposition \ref{propwellposlin}. Since  $A_m(v) \sim v^{m-1}$ (and bounded on compact sets) and $G$ is bounded, the term $B^{n,m}_t$ is bounded. Moreover, the sum of the terms in the square bracket in the definition of $R^{n,m}$ grows at infinity like  $v^{m}$.  So $R^{n,m}_t \sim v^m f^n_t$; therefore (applying Proposition \ref{propwellposlin} to \eqref{lineqniter}), $R^{n,m}_t \in L^{2,2}$ as soon as $f_0^n(x,v)=f_0(x,v)$ is such that $Y_0^{0,m}=\sqrt{1+v^{2m}} f_0 \in L^{2,2}$. If this is the case, then (see point (iii) of Remark \ref{rembeg}) \eqref{Yeq} is well posed in ${X}^2$. Moreover, if $Y_0^{0,m} \in L^{\infty}(\mathbb{T}_x\times \mathbb{R}_v)$,  by \eqref{maxprineqn} we have,
\[
\|Y^{n,m}_t\|_{L^\infty} \le  \|Y^{n,m}_0\|_{L^\infty} + \int_0^t\!\mathrm{d} s\,  \| R^{n,m}_s\|_{L^\infty}\,. 
\]
Now 
\[ 
| R^{n,m}_t|  \le  C \frac{\left\vert A(v)\right\vert}{\sqrt{1+v^{2m}}}|Y^{n,m}_t | +  C \frac{\left\vert v A_m(v)\right\vert}{\sqrt{1+v^{2m}}}|Y^{n,m}_t | + C \frac{\left\vert A_m'(v) \right\vert}{ \sqrt{1+v^{2m}}}|Y^{n,m}_t |\,. 
\]
Because $A_m(v) \sim v^{m-1}$ and $A_m'(v) \sim v^{m-2}$ (and they are both bounded on compact sets), all the functions of $v$ that multiply $|Y^{n,m}_t |$ in the addends on the right-hand side of the above, are indeed bounded functions. Therefore,
\[
\|Y^{n,m}_t \|_{L^\infty} \leq  \|Y^{n,m}_0\|_{L^\infty} + C \int_0^t\!\mathrm{d} s\, \|Y^{n,m}_s\|_{L^\infty}\,, 
\]
hence the claim is just  a consequence of Gronwall Lemma. 

\smallskip
\noindent \textit{Proof of (\ref{sest})}.  The equation satisfied by the (vector valued) function $Z^{n,m}_t$ reads,
\begin{equation}\label{eqforz}
\partial_t Z^{n,m}_t + v\partial_x Z^{n,m}_t + B^{n,m}_t \partial_v Z^{n,m}_t -\partial_v (vZ^{n,m}_t) - \sigma \partial_{vv} Z^{n,m}_t = \tilde R^{n,m}_t\,,
\end{equation}
where
\[
\begin{split}
\tilde R^{n,m}_t  & = \left[A_m(v)(B^{n,m}_t - v)  - \sigma A_m'(v)\right] \begin{pmatrix} \partial_x f^n_t \\ \partial_v f^n_t \end{pmatrix} \\ & \quad + \begin{pmatrix} - \partial_x (G(M^{n-1}(t)) & 0  \\ -1  &  1\end{pmatrix} \begin{pmatrix}\sqrt{1+v^{2m}}\partial_x f^n_t \\ \sqrt{1+v^{2m}}\partial_v f^n_t \end{pmatrix}\,.
\end{split}
\]
We want to act as before, i.e., we want to enforce conditions on $f_0$ such that \eqref{eqforz} is well posed and the maximum principle,  Lemma \ref{lemmamaxprinc}, can be applied. To this end, we start by observing that $G$ has bounded derivative, so $|\partial_x (G(M^{n-1}(t,x))| = |G'(M^{n-1}(t,x))  \partial_x M^{n-1}(t,x)| \le C |\partial_x M^{n-1}(t,x)|$. We now claim that 
\begin{equation}
\label{stgx}
\sup_{n}\sup_{s\in[0,t]} \|\partial_x M^{n-1}(s,x) \|_{L^\infty} < \infty \quad\forall\, t\ge 0\,.
\end{equation}
 Assuming for a moment that the above is true (we will prove it later), we have that the first component of $\tilde{R}_t^{n,m}$, $(\tilde{R}_t^{n,m})_1$, grows like 
$(\tilde{R}_t^{n,m})_1 \sim v^m \partial_x f^n_t$. In order for the right-hand side of \eqref{eqforz} to satisfy the assumption ii) of Proposition \ref{propwellposlin}, we therefore need to make sure that $v^m \partial_xf^n$ belongs to $(X^{\ell})^*$, for some $\ell\geq 1$. If $f_0 \in L^{2,m}$ then $v^m\partial_x f^n \in \mathcal{H}^{-1, m}_v$; still for the well-posedness of the equation for $(Z_t^{n,m})_1$, one needs to  assume that the initial datum $\sqrt{1+v^{2m}}\partial_x f_0^n$ is in $L^{2,\ell}$, and in order to be able to apply the maximum principle we need $\ell\geq 2$ and $\sqrt{1+v^{2m}}\partial_x f_0^n \in L^{\infty}$ as well. 
We now repeat the same exercise for the second component of equation \eqref{eqforz}. As we have just shown, under the assumptions listed so far, $\sqrt{1+v^{2m}} \partial_xf^n$ is bounded, therefore $(\tilde{R}_t^{n,m})_2 \sim v^m \partial_v f^n$.   If $f_0$ is such that $\sqrt{1+v^{2m}} \partial_v f_0 \in L^{2,2} \cap L^{\infty}$ and $f_0 \in L^{2,2m+1}$, then we have both well posedness (with $(\tilde{R}_t^{n,m})_2 \in L^{2,1}$) and we can apply the maximum principle. 

It remains to prove the claim \eqref{stgx}.  Consider the definition \eqref{Mn-1} of $M^{n-1}(t,x)$. When differentiating \eqref{Mn-1} with respect to $x$, the denominator of $\partial_xM^{n-1}$ is boun\-ded below away from zero since we assume $\varphi$ strictly positive.  As for the numerator, by the bounded convergence theorem  we can take the derivative under the integral, after observing that $| f_t^{n-1}(y,v) v \varphi'(x-y) |\le C| f_t^{n-1}(y,v)| \sqrt{1+v^2}$ and the right-hand side of the above is integrable (uniformly with respect to $n$) by \eqref{fest} and \eqref{interp}. Therefore,
\begin{equation}
\label{stgx1}
\partial_x M^{n-1}(t) = \frac{\langle w \rangle_{f^{n-1}_t,\varphi'}\langle 1 \rangle_{f^{n-1}_t,\varphi} - \langle w \rangle_{f^{n-1}_t,\varphi}\langle 1 \rangle_{f^{n-1}_t,\varphi'}}{\langle 1 \rangle_{f^{n-1}_t,\varphi}^2}\,,
\end{equation}
which is, in turn, bounded uniformly with respect to $n$. 

\smallskip
\noindent \textit{Proof of (\ref{dest})}. By definitions \eqref{yz}, 
\[
H^{n,m}_t = D Z^{n,m}_t - A_m(v) \begin{pmatrix} 0 & \partial_x f^n_t \\ 0 & \partial_v f^n_t \end{pmatrix}\,.
\]
In view of \eqref{sest}, the second term on the right-hand side is bounded on compact time intervals, uniformly with respect to $n$. Therefore, it remains to prove a bound like \eqref{dest} for the Jacobian matrix $W^{n,m}_t :=  D Z^{n,m}_t$. After long but straightforward computations, one finds that $W^{n,m}_t$ solves the equation,
\[
\partial_t W^{n,m}_t + v\partial_x W^{n,m}_t + B^{n,m}_t \partial_v W^{n,m}_t -\partial_v (vW^{n,m}_t) - \sigma \partial_{vv} W^{n,m}_t = D \tilde R^{n,m}_t + Q^{n,m}_t\,,
\]
where, using the short notation $(Z_1,Z_2)$ for the components of $Z^{n,m}_t$, 
\[
Q^{n,m}_t = \begin{pmatrix} \partial_xB^{n,m}_t\partial_vZ_1 & \partial_x Z_1 + \partial_v B^{n,m}_t\partial_v Z_1 - \partial_v Z_1 \\ \partial_xB^{n,m}_t\partial_vZ_2 & \partial_x Z_2 + \partial_v B^{n,m}_t\partial_v Z_2 - \partial_v Z_2 \end{pmatrix}\,.
\]

By arguing as in getting \eqref{stgx}, we deduce that also $\|\partial_{xx} (G(M^{n-1}(t))\|_{L^\infty}$ is bounded on compact time intervals, uniformly with respect to $n$. Therefore, from the definition of $B^{n,m}_t$ and $\tilde R^{n,m}_t$, it follows that the matrix valued function $D\tilde R^{n,m}_t + Q^{n,m}_t$ has entries that behave like linear combinations of $v^m \partial_x f^n$, $v^m \partial_{xx}f^n$, $v^m\partial_{vx}f^n$, $v^m \partial_{vv}f^n$ (with coefficients that are functions bounded on compact time intervals, uniformly with respect to $n$). All of these objects are in $(X^m)^*$; therefore, with reasonings similar to those detailed so far, the proof is complete.  
\end{proof}

\medskip
% The data information below will be filled by AIMS editorial staff
%Received xxxx 20xx; revised xxxx 20xx.
\medskip

\end{document}